\documentclass[12pt,reqno]{amsart}

\usepackage{times}
\usepackage{graphicx}
\usepackage{amsmath,amsthm,amsfonts,amscd,amssymb, MnSymbol}
\usepackage{pst-node}
\usepackage{pst-plot}

\usepackage[all]{xy}

\usepackage{tikz, braids}
\usetikzlibrary{decorations.pathreplacing, arrows}

\newtheorem{thm}{Theorem}[section]
\newtheorem{cor}[thm]{Corollary}
\newtheorem{defn}[thm]{Definition}
\newtheorem{lem}[thm]{Lemma}
\newtheorem{rem}[thm]{Remark}
\newtheorem{prop}[thm]{Proposition}
\newtheorem{ex}[thm]{Example}

\newtheorem{notation}[thm]{Notation}
{ \theoremstyle{remark} }

\numberwithin{thm}{section}
\numberwithin{equation}{section}

\newcommand{\mc}{\mathcal}

\newcommand{\norm}[1]{\left\Vert#1\right\Vert}

\newcommand{\la}{\langle}
\newcommand{\ra}{\rangle}

\newcommand{\Comp}{\mathbb{C}}

\newcommand{\D}{\mathcal{D}}
\newcommand{\g}{\mathbb{G}}

\newcommand{\n}{\mathbb{N}}

\newcommand{\tor}{\mathbb{T}}

\newcommand{\z}{\mathbb{Z}}

\global\long\def\tp{\mathop{\xymatrix{*+<.7ex>[o][F-]{\scriptstyle \top}}
 } }

\newcommand{\QG}{\mathbb{G}}

\newcommand{\ch}{\hat{c}}

\newcommand{\xv}{\mathbf{x}}

\newcommand{\yv}{\mathbf{y}}

\begin{document}

\title{Quantum channels with quantum group symmetry}

\author {Hun Hee Lee}
\address{Hun Hee Lee, Department of Mathematical Sciences and the Research Institute of Mathematics, Seoul National University, Gwanak-ro 1, Gwanak-gu, Seoul 08826, Republic of Korea}
\email{hunheelee@snu.ac.kr}

\author{Sang-Gyun Youn}
\address{Sang-Gyun Youn, 
Department of Mathematics Education, Seoul National University, 
GwanAkRo 1, Gwanak-Gu, Seoul 08826, South Korea}
\email{s.youn@snu.ac.kr }

\maketitle

\begin{abstract}

In this paper we will demonstrate that any compact quantum group can be used as symmetry groups for quantum channels, which leads us to the concept of covariant channels. We, then, unearth the structure of the convex set of covariant channels by identifying all extreme points under the assumption of multiplicity-free condition for the associated fusion rule, which provides a wide generalization of the results of \cite{MSD17}. The presence of quantum group symmetry contrast to the group symmetry will be highlighted in the examples of quantum permutation groups and $SU_q(2)$. In the latter example, we will see the necessity of the Heisenberg picture coming from the non-Kac type condition. This paper ends with the covariance with respect to projective representations, which leads us back to Weyl covariant channels and its fermionic analogue.


\end{abstract}

\section{Introduction}

Conservation of symmetry has been one of the central themes in quantum theory. Symmetries are often described by group actions and it is natural to be interested in quantum objects, such as quantum states, invariant under these actions. Since the quantum states are operators acting on Hilbert spaces we can immediately move to the representation theory of groups for the candidates of these actions. Recent developments of quantum information theory (shortly, QIT) lead us to focus more on quantum systems whose state space, or the underlying Hilbert space is finite dimensional, which means that the representations we are interested in are finite dimensional ones. This is why we usually consider compact groups, where the theory of associated finite dimensional representations is rich.
More precisely we are interested in a quantum state $\rho \in B(H)$ which is invariant under a finite dimensional unitary representation $\pi: G \to B(H)$ for a compact group $G$, i.e. $\pi(x)\rho\,\pi(x)^* =  \rho$, $\forall x\in G$.

The QIT point of view provides us another important class of quantum objects, namely quantum channels. In finite dimensional setting, quantum channels are (in the Schr{\" o}dinger picture) completely positive and trace-preserving (CPTP) maps between matrix algebras, where compact groups naturally act. We say that a quantum channel is covariant if this action is preserved.

The above two symmetric objects, namely invariant states and covariant channels
have been extensively studied (in \cite{DKS05,DFH06,KW09,MW09,MS14,DTW16,WTB17,H17a,H17b} and so on). One surprising feature is that the usual connection between the two objects, namely the Choi-Jamio{\l}kowski map (shortly, CJ-map) still serves as the bridge between two symmetries once we focus on bipartite states with the canonical choice of group invariance reflecting the bipartite structure. More precisely, we are interested in a bipartite quantum state $\rho \in B(H\otimes K)$ such that $(u(x)\otimes v(x))\rho \, (u(x)\otimes v(x))^* = \rho$, $\forall x\in G$
for some finite dimensional unitary representations $u: G\to B(H)$ and $v:G\to B(K)$ of a compact group $G$. We say that such $\rho$ is $G_{(u,v)}$-invariant. We are also interested in a quantum channel $\Phi: B(\bar{H}) \to B(K)$ such that
$\Phi(\bar{u}(x)X\bar{u}(x)^*) = v(x)\Phi(X)v(x)^*$, $\forall x\in G$, which we call $G_{(\bar{u},v)}$-covariant, where $\bar{u}$ is the the conjugate representation of $u$. When there is no possibility of confusion we simply say $G$-invariant and $G$-covariant, respectively.

The concept of $G$-invariance of bipartite states even goes back to Werner's 1989 paper \cite{W89} (introducing Werner states) and a detailed explanation for the general case is in \cite{VW01}. Some of the follow-up researches \cite{K02,H18} observed that $G_{(u,v)}$-invariance can be transferred to $G_{(\bar{u},v)}$-covariance of channels via the CJ-map for certain special cases. This correspondence gives us Werner-Holevo channels (a counterexample on Amosov, Holevo and Werner's conjecture, see \cite{AHW00, WH02}) from Werner states \cite{W89} and depolarizing channels from isotropic states \cite{VW01}. These classes of channels are of prime importance in QIT, but the representations behind them were limited to the fundamental representations (and their conjugates) of the Lie groups $U(n)$ and $O(n)$, whose structures are relatively easy. It was the paper by Nuwairan \cite{AN14} whose main focus was to reveal the structure of all irreducibly $SU(2)$-covariant channels using a more involved representation theory. Nuwairan introduced the class of EPOSIC channels, which turns out to be the set of extreme points of all irreducibly $SU(2)$-covariant channels since their images through the CJ-map are proved to correspond to the set of extreme points of $SU(2)$-invariant states from \cite{VW01}.


The idea of EPOSIC channels was to focus on the irreducible decomposition of the tensor product of two irreducible representations of $SU(2)$, which produces intertwining isometries to be used as the Stinespring isometries for the EPOSIC channels. Understanding the structure of the tensor decomposition is a fundamental issue in the representation theory of compact groups and the details of the aforementioned intertwining isometries are encoded in so-called the ``{\em Clebsch-Gordan coefficients}". In \cite{BC18,BCLY20} the authors lifted the idea of EPOSIC channels to the case of general compact groups and even to their quantum counterpart, namely the case of compact quantum groups and introduced the class of {\em Clebsch-Gordan channels} (shortly, CG-channels). The initial motivation of CG-channels was to provide a large class of channels with interesting properties, which are essentially different from the known ones. It has been pointed out in \cite{BCLY20} that certain CG-channels satisfy covariance with respect to the underlying quantum group $\QG$, but the usage of $\QG$-covariance has been limited to obtain the bistochastic property of the corresponding CG-channels, and the study on the general structures of $\QG$-covariant channels has not been pursued. Moreover, the associated invariant states have been clarified only for the special class of CG-channels called {\it Temperley-Lieb channels}, whose underlying quantum groups have the same fusion rule as $SU(2)$.

It is the main purpose of this paper to prove that the connection between two symmetries through the CJ-map hold in full generality even for the quantum group symmetry under the assumption of multiplicity-free tensor decomposition of the associated representations. This assumption allows us to provide a simple characterization of extremal, invariant bipartite states, which was already suggested in \cite{MSD17} for the special case of finite groups. The key result here is that those extremal bipartite states can successfully be traced back to quantum channels, which is a highly non-trivial fact since the range of CJ-map (upto a scaling factor) does not cover all states. Moreover, it turns out that those channels are exactly CG-channels, which makes them into fundamental building blocks for covariant channels with respect to quantum group symmetry.

Once we include quantum groups for the possible symmetry groups, it is natural to be interested in whether we could find genuine quantum phenomena different from the classical group case. The first such phenomenon comes with the example of the permutation group $S_n$ and the quantum permutation group $S^+_n$ for $n\ge 4$ in Section \ref{sec-free-general}.
Here, the set of all $S_n$-covariant quantum channels with respect to its standard representation forms a 3-simplex, whilst we get a 2-simplex for the corresponding $S^+_n$-covariance. The second quantum phenomenon arises when we consider the symmetry (quantum) group $SU_q(2)$, which is one of the best-known quantum groups \cite{W87a}.
In this case we have some obstacles coming from so-called the non-Kac type property of $SU_q(2)$,
which forces us to use
partial {\em quantum trace} instead of the usual partial trace in the Stinespring procedure. The use of quantum trace is necessary to secure covariance, but also leads us to another obstacle that the resulting map (which we call the {\em Clebsch-Gordan map} in Definition \ref{def-CG-map}) is not trace-preserving (shortly, TP) in general. Fortunately, the resulting map is still unital upto a scaling factor, which allows us to use the model of quantum channels in the Heisenberg picture, namely unital completely positive (UCP) maps. Recall that the two pictures, Heisenberg's and Schr{\" o}dinger's are known to be equivalent via trace duality. When the underlying quantum group is of Kac type (such as $S^+_n$ or classical compact groups), this additional difficulty disappears since the quantum trace is the same as the usual trace. 

This paper is organized as follows: We gather basics of CJ-map and the representation theory of compact quantum groups in Section 2. We analyze the structure of the space of $\g$-invariant operators and prove that $\g$-covariance of linear maps transfers to $\g$-invariance of their CJ-matrices via CJ-map in Section 3. In section 4 we introduce Clebsch-Gordan maps (shortly, CG-maps) as the key ingredients for establishing that the above transference result remains true when we restrict our focus on the case of $\QG$-covariant channels under the multiplicity-free assumption. Section 5 deals with the behavior of CG-maps in terms of trace duality when the underlying quantum group is of Kac type, which explains why we do not need to summon the Heisenberg picture in this case. The last section is devoted to various concrete examples highlighting the consequences of the previous results. In fact, we establish a general framework of determining the set of all $\QG$-covariant channels in two steps, namely, (1) finding irreducible components out of the associated tensor product representations, and (2) specifying the corresponding CG-maps. We will explore the cases of quantum permutation group $S_n^+$ and the $q$-deformed quantum group $SU_q(2)$ exhibiting genuine quantum phenomena. The last collection of examples are covariance coming from projective representations. While these examples can be regarded as sub-cases of classical group symmetry, we would like to emphasize its connection to fundamental quantum systems such as finite Weyl systems and fermionic systems. This will lead us back to the well-known Weyl covariant channels and its fermionic analogue.

\section{Preliminaries}

\subsection{Choi-Jamio{\l}kowski map and trace duality}\label{subsec-Choi-adjoint}

In this paper all Hilbert spaces (denoted by $H_A$, $H_B$ and so on) are finite dimensional.

One of the fundamental object we will see in this paper is a linear map $\Phi:B(H_A) \rightarrow B(H_B)$. For a fixed choice of orthonormal basis $\{e_i\}^n_{i=1}$, $n = {\rm dim}(H_A)$ we consider the following linear isomorphism
    $$C: B(B(H_A),B(H_B)) \to B(\bar{H}_A\otimes H_B) \cong B(\bar{H}_A)\otimes B(H_B),\;\; \Phi \mapsto C_\Phi,$$
where
    $$C_{\Phi}= \sum_{i,j=1}^{n} e_{ij}\otimes \Phi(e_{ij}).$$
Here, $e_{ij}$ is the matrix unit associated with the fixed basis $\{e_i\}^n_{i=1}$ and $\bar{H}_A$ refers to the Hilbert space conjugate of $H_A$. We call the map $C$ and the operator $C_\Phi$ by the Choi-Jamio{\l}kowski map (shortly, CJ-map) and the CJ-matrix of $\Phi$, respectively. It is straightforward to check that
    \begin{align}\label{eq-CJ}
        \Phi(X) = ({\rm Tr}\otimes {\rm id})(C_\Phi(X^t \otimes {\rm Id})),\;\; X\in B(H_A),
    \end{align}
where $X^t$ refers to the transpose of $X$.    

The CJ-map gives a bijection between $CP(B(H_A),B(H_B))$, the set of all completely positive (shortly, CP) maps and the set of all positive definite matrices in $B(H_A\otimes H_B)$. Since $\Phi: B(H_A)\to B(H_B)$ is trace-preserving if and only if $({\rm id }\otimes {\rm Tr})(C_{\Phi})={\rm Id}_A$, which is a stronger condition than ${\rm Tr}(C_{\Phi})={\rm dim}(H_A)$, the map $\Phi\mapsto \displaystyle \frac{1}{{\rm dim}(H_A)}C_{\Phi}$ is an injective but not surjective mapping from $CPTP(B(H_A),B(H_B))$, the set of all CPTP maps into $\mathcal{D}(H_A\otimes H_B)$, the set of all states on $H_A\otimes H_B$. This is often called the {\it channel-state duality}.



The transition from the Schr{\" o}dinger picture into the Heisenberg picture is done via trace duality, which is based on the following two natural duality brackets on the matrix algebra $B(H)$. 
    $$\la X,Y\ra := {\rm Tr}(XY)\;\; \text{and}\;\; \llangle X,Y\rrangle := {\rm Tr}(X^tY),\; X,Y\in  B(H).$$

For a linear map $\Phi: B(H_A) \to B(H_B)$ we now have two associated adjoint maps $\Phi^*$ and $\Phi'$ given by $\la \Phi^*(Y), X\ra = \la Y, \Phi(X) \ra$ and $\llangle \Phi'(Y), X\rrangle = \llangle Y, \Phi(X)\rrangle$, respectively for $X\in B(H_A)$ and $Y\in B(H_B)$. We can easily check that
    $$\Phi'(X) = \Phi^*(X^t)^t,\;\; X\in B(H_B).$$

\subsection{Compact quantum group}

A compact group $G$ can be understood by the pair $(C(G), \Delta)$, where $C(G)$ is the space of all complex valued continuous functions on $G$ and $\Delta: C(G) \to C(G\times G)$ is the map given by $[\Delta (f)](s,t) = f(st)$, $s,t\in G$, encoding the group multiplication. There is a unique Borel probability measure $\mu$ on $G$ called the {\em Haar measure}, which is translation invariant. This produces a positive functional (called the {\em Haar functional}) on $C(G)$ given by $f \mapsto \int_G f\, d\mu$. This style of describing compact groups has a quantum counterpart leading us to the {\em compact quantum group}, which is given by a pair $(C(\g),\Delta)$. In this case $C(\g)$ is a unital $C^*$-algebra and $\Delta:C(\g)\rightarrow C(\g)\otimes C(\g)$ is a unital $*$-homomorphism, called the {\em co-multiplication}, satisfying the following properties:
\begin{enumerate}
    \item $(\Delta \otimes {\rm id})\Delta=({\rm id}\otimes \Delta)\Delta$.
    \item Both $\left \{ \Delta(a)(b\otimes 1):a,b\in C(\g) \right \}$ and $\left \{ \Delta(a)( 1 \otimes b ):a,b\in C(\g) \right \}$ span dense subspaces in $C(\g)\otimes  C(\g)$.
\end{enumerate}
Note that the minimal tensor product of $C^*$-algebras is used for the space $C(\g)\otimes C(\g)$.

There is a unique, unital, positive, linear functional $h:C(\g)\rightarrow \Comp$ satisfying
\[({\rm id}\otimes h)\Delta = h(\cdot )1 = (h\otimes {\rm id})\Delta,\]
which replaces the role of the Haar functional and is called the {\it Haar state}. See \cite{W87a,W87b,T08} for more details.

A (finite dimensional) {\em representation} of $\g$ is an invertible element
$$u= (u_{ij})^{n_u}_{i,j=1}= \sum_{i,j=1}^{n_u} e_{ij}\otimes u_{ij} \in B(H_u)\otimes C(\g)\cong M_{n_u}(C(\g))$$ 
satisfying
$$\Delta(u_{ij})=\sum_{k=1}^{n_u}u_{ik}\otimes u_{kj}\text{ for all }1\leq i,j\leq n_u.$$
Here, $n_u={\rm dim}(H_u)$ is the {\it (classical) dimension} of $u$. When the element $u$ is unitary, we say that $u$ is a {\em unitary representation}.

For a unitary representation $u$, we consider the {\it contragradient representation} $u^c$ given by 
$$u^c= (u_{ij}^*)_{i,j} =\displaystyle \sum_{i,j=1}^{n_u} e_{ij}\otimes u^*_{ij} \in B(\bar{H}_u)\otimes C(\g) \cong M_{n_u}(C(\QG)),$$
which is, in general, not unitary. However, there is a uniquely determined invertible positive matrix $Q_u \in B(\bar{H}_u)$ such that the functionals ${\rm Tr}(Q_u \, \cdot )$ and ${\rm Tr}(Q^{-1}_u \, \cdot)$ coincides on $B(\bar{H}_u)$ and the element
    $$\overline{u} = (Q^{\frac{1}{2}}_{u} \otimes 1)u^c(Q^{-\frac{1}{2}}_{u} \otimes 1)$$
is a unitary representation, which we call the {\em conjugate representation}. Here, the common value ${\rm Tr}(Q_u) = {\rm Tr}(Q^{-1}_u) = d_u$ is called the {\em quantum dimension} of $u$. For $u^t = (u_{ji})_{i,j}$ we have
    \begin{equation}\label{eq-contra-transpose}
    u^t(Q_{u} \otimes 1) u^c = Q_{u} \otimes 1,\;\; u^c(Q_{u}^{-1} \otimes 1)u^t  = Q^{-1}_{u} \otimes 1.    
    \end{equation}
From this, it is not difficult to check $Q_{\overline{u}} = Q^{-1}_{u}$. The matrix $Q_u$ can be associated with the positive functional ${\rm Tr}_{Q_u}$ on $B(H_u)$ given by
    $${\rm Tr}_{Q_u}(X) = {\rm Tr}(Q_u X), \; X\in B(H_u),$$
which we call a {\em quantum trace}. We say that $\QG$ is of {\em Kac type} if $Q_u = {\rm Id}_u$ for any unitary representation $u$ of $\QG$.

For two representations $u = (u_{ij})$ and $v = (v_{kl})$ of $\QG$ we define the {\em tensor product} $u \tp v$ by
    $$u \tp v := \sum_{i,j=1}^{n_u}\sum_{k,l=1}^{n_v}e_{ij}\otimes e_{kl}\otimes u_{ij}v_{kl}\in B(H_u)\otimes B(H_v)\otimes C(\g),$$
and the {\em direct sum} $u \oplus v \in (B(H_u) \oplus B(H_v))\otimes C(\g)$ in an obvious way.

We also consider the space of {\em intertwiners}
    $$\text{Hom}(u,v) := \{A\in B(H_u, H_v): v(A\otimes 1)=(A\otimes 1)u\}.$$
We say that $u$ and $v$ are {\em equivalent} (we write $u \cong v$) if there is an invertible operator in $\text{Hom}(u,v)$. We say that $u$ is {\it irreducible} if $\text{Hom}(u,u) = \Comp\cdot {\rm Id}_u$. The set of equivalence classes of irreducible unitary representations of $\QG$ will be denoted by ${\rm Irr}(\QG)$. We will always assume that for a given $\alpha \in {\rm Irr}(\QG)$ we fix a representative $u^\alpha \in B(H_\alpha) \otimes C(\QG)$. In this case, $Q_{u^\alpha}$, $n_{u^\alpha}$ and $d_{u^\alpha}$ will be simply denoted by $Q_\alpha$, $n_\alpha$ and $d_\alpha$, respectively.

We say that $u$ is a {\em subrepresentation} of $v$ (we write $u \subseteq v$) if there is another representation $w$ on $\QG$ such that $u \oplus w$ is equivalent to $v$.

Irreducible unitary representations play a vital role in understanding the underlying quantum group $\QG$. The linear space ${\rm Pol}(\g)$ spanned by $\{u^\alpha_{ij}: \alpha \in {\rm Irr}(\QG),\; 1\le i,j\le n_\alpha \}$ is dense in $C(\QG)$ and equipped with the {\it antipode} $S:{\rm Pol}(\g)\rightarrow {\rm Pol}(\g)$ given by 
    $$S(u^{\alpha}_{ij})=(u^{\alpha}_{ji})^*,\;\; \alpha\in {\rm Irr} (\g),\; 1\leq i,j\leq n_{\alpha},$$
which is an anti-multiplicative linear map. Moreover, we have the following {\it Schur-orthogonality} relations: for $\alpha,\beta \in {\rm Irr}(\QG)$ we have
    \begin{equation}
        h(u^\alpha_{ij}(u^\beta_{kl})^*) = \delta_{\alpha\beta}\delta_{ik}\frac{(Q_\alpha)_{lj}}{d_\alpha}\text{ and }h((u^{\alpha}_{ij})^*u^{\beta}_{kl})=\delta_{\alpha\beta}\delta_{jl}\frac{(Q_{\alpha}^{-1})_{ki}}{d_{\alpha}}.
    \end{equation}
Here, by choosing a suitable orthonormal basis of $H_{\alpha}$, we may assume that $Q_{\alpha}$ is diagonal \cite{D10}.

Every (finite dimensional) unitary representation of $\QG$ can be decomposed into a direct sum of irreducible ones. In this paper, we are interested in the irreducible decomposition of the tensor product of two unitary representations $u,v$ of $\QG$, namely $u\tp v \cong \oplus^N_{j=1} u^{\alpha_j}$. This decomposition, in general, allows repetition of the same representatives, i.e. you may have $\alpha_j = \alpha_k$ for $1\le j\ne k \le N$. We say that $u\tp v$ has a {\em multiplicity-free irreducibel decomposition} if all the elements $\alpha_j$, $1\le j\le N$ are distinct. We are particularly interested in the case when $u$ and $v$ are also irreducible, namely $u = u^\alpha$, $v = u^\beta$, $\alpha,\beta \in {\rm Irr}(\QG)$. Then we have $u^\gamma \subseteq u^\alpha \tp u^\beta$ for $\gamma = \alpha_j$, $1\le j\le N$, which we simply write
    $$\gamma \subseteq \alpha \tp \beta.$$
It is well-known that the above is equivalent to (see \cite[Section 3.1.3]{T08} for the details)
    $$\overline{\gamma}\subseteq \overline{\beta}\tp \overline{\alpha} \Leftrightarrow \beta \subseteq \gamma \tp \overline{\alpha}\Leftrightarrow \beta\subseteq \overline{\alpha}\tp \gamma \Leftrightarrow \overline{\beta} \subseteq \overline{\gamma} \tp \alpha.$$
In this case, we can find an isometric intertwiner $v^{\alpha,\beta}_\gamma \in {\rm Hom}(\gamma, \alpha \tp \beta)$, whose behavior with respect to the $Q$-matrices are given as follows (\cite[Remark 4.6]{FLS16}).
    \begin{equation}\label{Q-matrix-v}
        (Q_{\alpha} \otimes Q_{\beta}) v^{\alpha, \beta}_\gamma = v^{\alpha, \beta}_\gamma Q_\gamma.
    \end{equation}

\subsection{The quantum permutation group $S^+_n$ and the $q$-deformed quantum group $SU_q(2)$}

A compact quantum group $(C(\QG), \Delta)$ is called a compact matrix quantum group (see \cite[Proposition 6.1.4 i)]{T08} for the details) if there is a unitary $w=(w_{ij})_{1\leq i,j\leq n} \in M_n(C(\QG))$, called the {\em fundamental representation}, such that
\begin{enumerate}
    \item $w^c=(w^*_{ij})_{1\leq i,j\leq n}$ is invertible,
    \item the elements $w_{ij}$, $1\leq i,j\leq n$, generate $C(\g)$ and
    \item $\Delta(w_{ij})=\displaystyle \sum_{k=1}^n w_{ik}\otimes w_{kj}$, $1\leq i,j\leq n$,
\end{enumerate}
which can be regarded as a quantum analogue of compact Lie groups.  Let us collect basic materials for the compact matrix quantum groups $S^+_n$ and $SU_q(2)$, which will be used only in Section \ref{sec-free-general} and Section \ref{sec-SUq(2)}. 


\begin{ex}
The quantum permutation group $S_n^+$ ($n\geq 2$) was introduced in \cite{W98}. The underlying $C^*$-algebra $C(S_n^+)$ is given by the universal unital $C^*$-algebra generated by $n^2$ operators $u_{ij}$ satisfying 
\begin{enumerate}
    \item $w_{ij}^*=w_{ij}=w_{ij}^2$ for all $1\leq i,j\leq n$ and
    \item $\displaystyle \sum_{i=1}^n w_{ij}=1$ for all $1\leq i\leq n$ and $\displaystyle \sum_{j=1}^n w_{ij}=1$ for all $1\leq j\leq n$.
\end{enumerate} 
Then the unital $*$-homomorphism $\Delta$ determined by $w_{ij}\mapsto \displaystyle \sum_{k=1}^n w_{ik}\otimes w_{kj}$, $1\leq i,j\leq n$, turns $(C(S_n^+),\Delta)$ into a compact quantum group.

It is known \cite{B99} that ${\rm Irr}(S_n^+)$ is identified with $\left \{0,1,2,\cdots \right\}$ and 
    $$u^l\tp u^m \cong u^{l+m}\oplus u^{l+m-1}\oplus \cdots \oplus u^{|l-m|}.$$
Note that the {\it fundamental representation} $w=(w_{ij})_{1\leq i,j\leq n}$ is reducible with two invariant subspaces $H_0=\Comp\cdot \xi_0$ and $H_1=H_0^{\perp}$ for the invariant vector $\xi_0=\sum_{j=1}^n e_j$. The associated irreducible sub-representations are known to be equivalent to $u^0$ and $u^1$, respectively.
\end{ex}
 
 \begin{ex}
 The compact quantum group $SU_q(2)$ ($0<q<1$) was introduced in \cite{W87a}. The underlying $C^*$-algebra $C(SU_q(2))$ is the universal unital $C^*$-algebra generated by operators $a,c$ such that $\left [ \begin{array}{cc} a&-qc^*\\ c&a \end{array} \right ]$ is a unitary. Then, together with the unital $*$-homomorphism $\Delta$ determined by $\begin{array}{ll}a\mapsto& a\otimes a-qc^*\otimes c\\
 c\mapsto&c\otimes a+a^*\otimes c \end{array}$, the pair $(C(SU_q(2)),\Delta)$ satisfies the axioms to be a compact quantum group. It is known that ${\rm Irr}(SU_q(2))$ is identified with $\left \{0,1,2,\cdots\right\}$ with the fusion rule
 $$u^l\tp u^m\cong u^{l+m}\oplus u^{l+m-2}\oplus \cdots \oplus u^{|l-m|}.$$
 Note that $u^0 = 1$, the identity of the algebra $C(SU_q(2))$ and $u^1 = \left [ \begin{array}{cc} a&-qc^*\\ c&a \end{array} \right ] \in M_2(C(SU_q(2)))$, the fundamental representation.
 \end{ex}

\section{$\QG$-invariant operators and $\QG$-covariant linear maps}

We begin with $\QG$-invariant states or more generally $\QG$-invariant operators for a compact quantum group $\QG$.

\begin{defn} Let $u$ and $v$ be finite dimensional unitary representations of $\QG$.
An operator $X\in B(H_u)$ is called $\QG_u$-invariant if $X$ is an intertwiner of $u$, i.e. 
    $$u(X\otimes 1)u^*=X\otimes 1.$$
In particular, a bipartite operator $Y\in B(H_u\otimes H_v)$ is called $\QG_{(u,v)}$-invariant if
    \begin{equation}\label{eq-bipartite-inv}
      (u\tp v)(Y\otimes 1)(u \tp v)^*=Y\otimes 1.  
    \end{equation}
When $u,v$ are irreducible, i.e. $u=u^\alpha$ and $v=u^{\beta}$ for some $\alpha,\beta \in {\rm Irr}(\QG)$, we write ``$\QG_\alpha$-invariant" instead of ``$\QG_{u^\alpha}$-invariant" and ``$\g_{(\alpha,\beta)}$-invariant'' instead of ``$\g_{(u^{\alpha},u^{\beta})}$-invariant''.
\end{defn}

We will focus on $\QG_{(u,v)}$-invariance under the condition that $u \tp v$ has a multiplicity-free irreducible decomposition, which makes their structure easy to describe as we can see below.

\begin{prop}\label{prop-intertwiner}
Suppose that $u_1,u_2,\cdots, u_n$ are mutually inequivalent irreducible unitary representations of $\QG$ and set $u = \oplus_{j=1}^n u_j$. Then, $X\in B(H_u)$ is $\QG_u$-invariant if and only if $X\in {\rm span}\left \{ p_k \right\}_{1\leq k\leq n}$, where $p_k$ is the orthogonal projection onto the space $H_{u_k}$ from $\bigoplus_{k=1}^n H_{u_k}$. Moreover, the set of extreme points of the convex set of $\QG_u$-invariant states $\D(\QG_u)$ is
    $${\rm Ext}\D(\QG_u) = \left \{\frac{1}{n_k}p_k: 1\leq k\leq n\right \}.$$
\end{prop}
\begin{proof}
Let us denote by $e^{k}_{ij}$ the canonical matrix units of $B(H_{u_k})$, which allow us to write $\oplus_{k=1}^n u_k$ as $\displaystyle \sum_{k=1}^n \sum_{i,j=1}^{n_k} e^{k}_{ij}\otimes u^k_{ij}$. Then the intertwining property of $X\otimes 1$ becomes
    \begin{equation}\label{eq-intertwining}
        X\otimes 1=\sum_{k_1,k_2=1}^n \sum_{i,j=1}^{n_{k_1}}\sum_{k,l=1}^{n_{k_2}} e^{k_1}_{ij}Xe^{k_2}_{lk}\otimes u^{k_1}_{ij}(u^{k_2}_{kl})^*.
    \end{equation}

Now we assume that $X\otimes 1$ is an intertwiner of $\oplus_{j=1}^n u_j$. Taking ${\rm id}\otimes h$ on both sides of \eqref{eq-intertwining}, we obtain 
\begin{align*}
X&=\sum_{k=1}^n \sum_{i,j=1}^{n_k} X^k_{jj} e^k_{ii} \frac{(Q_{u_k})_{jj}}{d_k}=\sum_{k=1}^n \frac{{\rm Tr}(XQ_{u_k})}{d_k}p_k.
\end{align*}
by the Schur's orthogonality relation. The converse direction is obtained directly by the unitarity of $u_k$, $1\le k \le n$.

The last statement for states follows directly from the fact that the projections $p_k$, $1\le k \le n$, have orthogonal ranges, so that a state $\rho\in B(H_u)$ is $\QG_u$-invariant if and only if $\rho = \displaystyle \sum^n_{k=1}\frac{a_k}{n_k} p_k$ for some probability distribution $\{a_k\}^n_{k=1}$.
\end{proof}


We continue with $\QG$-covariant linear maps and their basic properties.
    \begin{defn}
    Let $u,v$ be finite dimensional unitary representations of $\QG$. We say that a linear map $\Phi: B(H_u) \to B(H_v)$ is $\QG_{(u,v)}$-covariant if
        $$(\Phi \otimes {\rm id})[u(X\otimes 1)u^*] = v(\Phi(X)\otimes 1)v^*,\;\; X\in B(H_u).$$
    \end{defn}

\begin{notation}
\begin{enumerate}
\item ${\rm Cov}_\QG(u,v)$, the space of all $\QG_{(u,v)}$-covariant maps.
\item ${\rm CPTPCov}_\QG(u,v)$, the convex set of all CPTP $\QG_{(u,v)}$-covariant maps.
\item ${\rm UCPCov}_\QG(u,v)$, the convex set of all UCP $\QG_{(u,v)}$-covariant maps.
\end{enumerate}
\end{notation}

We can immediately connect two symmetries via the CJ-map.

\begin{thm}\label{thm-cov-Choi}
Let $\Phi:B(H_u)\rightarrow B(H_v)$ be a linear map for finite dimensional unitary representations $u,v$ of $\QG$. Then $\Phi$ is $\QG_{(u,v)}$-covariant if and only if
\begin{equation}\label{eq2}
(u^c\tp v)^*(C_{\Phi}\otimes 1) (u^c\tp v)  = C_{\Phi}\otimes 1.
\end{equation}
In other words, $\Phi$ is $\QG_{(u,v)}$-covariant if and only if the operator $(Q_u^{-\frac{1}{2}}\otimes {\rm Id})C_{\Phi}(Q_u^{-\frac{1}{2}}\otimes {\rm Id})$ is $\QG_{(\overline{u},v)}$-invariant.
\end{thm}

\begin{proof}
For any $X\in B(H_u)$, let us denote by $\varphi_X$ the linear functional given by $\varphi_X(A) = {\rm Tr}(A X^t)$. Then we have
\begin{align*}
&v^*(\Phi\otimes {\rm id})(u(X \otimes 1)u^*)v\\
&=\sum_{i,j,k,l=1}^{n_u} \sum_{p,q,r,s=1}^{n_v}e^v_{qp}\Phi(e^u_{ij} X e^u_{lk})e^v_{rs}\otimes v_{pq}^*u_{ij}u_{kl}^* v_{rs}\\
&=\sum e^v_{qp}(\varphi_{e^u_{ij}X e^u_{lk}}\otimes {\rm id})(C_{\Phi} )e^v_{rs}\otimes v_{pq}^*u_{ij}u_{kl}^* v_{rs}\\
&=\sum e^v_{qp}(\varphi_X\otimes {\rm id})((e^u_{ji}\otimes {\rm Id}_{v})C_{\Phi}(e^u_{kl}\otimes {\rm Id}_{v}) )e^v_{rs}\otimes v_{pq}^*u_{ij}u_{kl}^* v_{rs}\\
&=\sum (\varphi_X \otimes {\rm id})[(e^u_{ji}\otimes e^v_{qp})C_{\Phi} (e^u_{kl} \otimes e^v_{rs})]\otimes  v_{pq}^*u_{ij}u_{kl}^* v_{rs}\\
&=(\varphi_X\otimes {\rm id}\otimes {\rm id})((u^c\tp v)^* (C_{\Phi}\otimes 1) (u^c\tp v)).
\end{align*}
On the other hand we have $\Phi(X)\otimes 1= (\varphi_X \otimes {\rm id}\otimes {\rm id})(C_{\Phi}\otimes 1)$. Thus, we can see that $\QG_{(u,v)}$-covariance of $\Phi$ is the same as $(\varphi_X \otimes {\rm id}\otimes {\rm id})(C_{\Phi}\otimes 1) = v^*(\Phi\otimes {\rm id})(u(X \otimes 1)u^*)v$ for any $X \in B(H_u)$, which, in turn, is equivalent to the identity
\[(u^c\tp v)^* (C_{\Phi}\otimes 1) (u^c\tp v)= C_{\Phi}\otimes 1 .\]

The last conclusion immediately follows from the following identity
$$u^c\tp v = (Q_u^{-\frac{1}{2}}\otimes {\rm Id}_v\otimes 1)(\overline{u}\tp v)(Q_u^{\frac{1}{2}}\otimes {\rm Id}_v\otimes 1).$$

\end{proof}

\begin{rem}\label{rmk-surjectivity}

A new perspective on the channel-state duality comes from Theorem \ref{thm-cov-Choi}, which states that $\Phi\mapsto \frac{1}{{\rm dim}(H_A)}C_{\Phi}$ is an injective mapping from ${\rm CPTPCov}_{\g(u,v)}$ into $\displaystyle \mathcal{D}(\g_{(\bar u,v)})$ if $\g$ is of Kac type. A surprising feature is that we even get surjectivity of this correspondence under the multiplicty-free assumption, which we will endeavor in the next section.
\end{rem}

When the associated unitaries $u$ and $v$ are irreducible, i.e. $u=u^\alpha$ and $v=u^\beta$ for some $\alpha,\beta \in {\rm Irr}(\QG)$, then we will simply write $\QG_{(\alpha,\beta)}$-covariance instead of $\QG_{(u^\alpha,u^\beta)}$-covariance. In this case a linear map $\Phi: B(H_\alpha)\to B(H_\beta)$ being $\g_{(\alpha,\beta)}$-covariant clearly implies that $\Phi({\rm Id}_{\alpha})$ is $\g_{\beta}$-invariant, so we obtain the following.
\begin{prop}\label{prop-cov-unital}
Let $\alpha,\beta \in {\rm Irr}(\QG)$ and $\Phi$ be a $\QG_{(\alpha,\beta)}$-covariant map. Then we have
\[\Phi({\rm Id}_{\alpha})\in \Comp \cdot {\rm Id}_{\beta}.\]
\end{prop}

\section{Clebsch-Gordan maps and the structure of $\QG$-covariant maps}\label{sec-CG}


Looking back Proposition \ref{prop-intertwiner} we can immediately see that the operators $p_k$, the orthogonal projection onto the space $H_{u_k}$ from $\bigoplus_{k=1}^n H_{u_k} \cong \overline{u}\tp v$, are building blocks for $\QG_{(\overline{u}, v)}$-invariant operators. The main result of this section is the construction of linear maps whose CJ-matrices are exactly the operators $p_k$.

\begin{defn}\label{def-CG-map}
Let $\alpha,\beta,\gamma \in {\rm Irr}(\g)$ such that $\alpha\subseteq \beta\tp \gamma$. We define the linear map $\Phi^{\alpha \to \beta}_\gamma: B(H_\alpha) \to B(H_\beta)$ given by
    $$\Phi^{\alpha \to \beta}_\gamma(A) :=({\rm id}\otimes {\rm Tr}_{Q_{\gamma}})(v^{\beta,\gamma}_{\alpha} A (v^{\beta,\gamma}_{\alpha})^*),\;\; A\in B(H_\alpha).$$
We call the maps $\Phi^{\alpha \to \beta}_\gamma$ {\em Clebsch-Gordan map}s (shortly, CG-maps) on $\QG$.
\end{defn}

\begin{rem}
\begin{enumerate}
    \item The above construction was already considered to study irreducibly $SU(2)$-covariant channels in \cite{L78,LS14,AN14}.
    \item In \cite{BCLY20}, the usual trace ${\rm Tr}$ was used instead of the quantum trace ${\rm Tr}_{Q_{\gamma}}$ to construct {\it Clebsch-Gordan channels}. These maps are indeed quantum channels, but not $\g$-covariant in general. See Remark \ref{rmk-quantum-trace} for more details.
    \item The CG-map $\Phi^{\alpha\rightarrow \beta}_{\gamma}$ coincides with $\Phi^{\beta,\overline{\gamma}}_{\alpha}$ in \cite{BCLY20} if $\g$ is of Kac type, but our notation has the merit of visualizing the connection with its CJ-matrix and avoiding some confusions such as $\overline{\overline{\gamma}}$.
\end{enumerate}
\end{rem}

We collect some basic properties of Clebsch-Gordan maps below.
\begin{prop}\label{prop-CG-map-cov}
Suppose that $\alpha,\beta,\gamma\in {\rm Irr}(\g)$ with $\alpha\subseteq \beta\tp \gamma$.
    \begin{enumerate}
        \item The map $\Phi^{\alpha \to \beta}_\gamma$ is $\QG_{(\alpha,\beta)}$-covariant.
        \item The map $\Phi^{\alpha \to \beta}_\gamma$ is quantum trace preserving, i.e. ${\rm Tr}_{Q_\beta}(\Phi^{\alpha \to \beta}_\gamma(X)) = {\rm Tr}_{Q_\alpha}(X)$, $X\in B(H_\alpha)$.
        \item We have $\Phi^{\alpha \to \beta}_\gamma({\rm Id}_\alpha) = \frac{d_\alpha}{d_\beta}{\rm Id}_\beta$.
    \end{enumerate}
        
\end{prop}

\begin{proof}
(1) Let us write $v^{\beta,\gamma}_{\alpha} = v$ for simplicity. For any $X\in B(H_u)$, we have
\begin{align*}
& (\Phi^{\alpha \to \beta}_\gamma\otimes {\rm id})(u^{\alpha} (X \otimes 1)(u^{\alpha})^* )\\
&=({\rm id}\otimes {\rm Tr}_{Q_{\gamma}}\otimes {\rm id})[(v \otimes {\rm Id}) u^{\alpha}(X\otimes 1)(u^{\alpha})^*( v \otimes {\rm Id})^*]\\
&=({\rm id}\otimes {\rm Tr}_{Q_{\gamma}}\otimes {\rm id})[u^{\beta}\tp u^{\gamma}   (vXv^*\otimes 1) (u^{\beta}\tp u^{\gamma})^*]\\
&=\sum^{n_\beta}_{i,j,p,q=1}\sum^{n_\gamma}_{k,l,r,s=1} ({\rm id}\otimes {\rm Tr}_{Q_{\gamma}})((e^\beta_{ij}\otimes e^\gamma_{kl})vXv^* (e^\beta_{qp}\otimes e^\gamma_{sr})) \otimes u^{\beta}_{ij}u^{\gamma}_{kl}(u^{\gamma}_{rs})^*(u^{\beta}_{pq})^*\\
&= \sum_{ \substack{i,j,p,q \\ l,s}} e^\beta_{ij} ({\rm id}\otimes {\rm Tr})(({\rm Id}\otimes e^\gamma_{sl})v X v^* ) e^\beta_{qp} \otimes  u^{\beta}_{ij}[(u^{\gamma})^t (Q_\gamma\otimes 1)(u^{\gamma})^c]_{ls}(u^{\beta}_{pq})^*\\
&=\sum_{i,j,p,q} e^\beta_{ij}({\rm id}\otimes {\rm Tr}_{Q_\gamma})(vXv^*)e^\beta_{qp}\otimes  u^{\beta}_{ij}(u^{\beta}_{pq})^*\\
&=u^{\beta} (\Phi^{\alpha \to \beta}_\gamma\otimes {\rm id})(X\otimes 1) (u^{\beta})^*.
\end{align*}
On the second last equality, we used the fact that $(u^{\gamma})^t (Q_\gamma\otimes 1)(u^{\gamma})^c = Q_\gamma\otimes 1$.

\vspace{1cm}

(2) For $X\in B(H_\alpha)$ we have
    \begin{align*}
        {\rm Tr}_{Q_\beta}(\Phi^{\alpha \to \beta}_\gamma(X))
        &= ({\rm Tr}\otimes {\rm Tr})[(Q_\beta \otimes Q_\gamma)v^{\beta,\gamma}_\alpha X (v^{\beta,\gamma}_\alpha)^*]\\
        &= {\rm Tr}\left ( X (v^{\beta,\gamma}_\alpha)^*(Q_\beta \otimes Q_\gamma)v^{\beta,\gamma}_\alpha \right )\\
        & = {\rm Tr}\left (X (v^{\beta,\gamma}_\alpha)^*v^{\beta,\gamma}_\alpha Q_\alpha \right ) = {\rm Tr}_{Q_\alpha}(X),
    \end{align*}
where we used \eqref{Q-matrix-v} in the third equality.

\vspace{1cm}

(3) This is immediate from the above result and Proposition \ref{prop-cov-unital}.
\end{proof}

From the identity \eqref{eq-contra-transpose} we immediately get the following.

\begin{lem}\label{lem1}
\begin{enumerate}
\item For any $\alpha\in {\rm Irr}(\g)$ and $X\in B(H_{\alpha})$, we have
\begin{equation}
({\rm Tr}_{Q_{\alpha}}\otimes {\rm id})(u^{\alpha}(X\otimes 1)(u^{\alpha})^*)={\rm Tr}_{Q_{\alpha}}(X)\otimes 1.
\end{equation}
\item For any $\alpha\in {\rm Irr}(\g)$ we have
\begin{equation}
((u^{\alpha})^c\tp u^{\alpha})\left (\sum_{j=1}^{n_{\alpha}}(Q_{\alpha}^{-1})_{jj}|j\ra\otimes |j\ra \otimes 1 \right )=\sum_{j=1}^{n_{\alpha}}(Q_{\alpha}^{-1})_{jj}|j\ra\otimes |j\ra \otimes 1.
\end{equation}
\end{enumerate}
\end{lem}

Now we are ready to determine the CJ-matrices of Clebsch-Gordan maps with multiplicity-free condition. The case of $SU(2)$ was proved in \cite[Proposition 4.5]{AN14} based on a detailed analysis of Clebsch-Gordan coefficients of $SU(2)$, which seems available only in rare cases. The case of free orthogonal quantum groups $O_N^+$ was proved in \cite[Theorem 3.3]{BCLY20} using diagrammatic calculus. 
Both proofs are quite different from our approach.

\begin{thm}\label{thm-Choi}
Suppose that $\overline{\alpha}\tp \beta$ has a multiplicity-free irreducible decomposition. Then, for any $\overline{\gamma} \subseteq \overline{\alpha}\tp \beta$, $\gamma \in {\rm Irr}(\QG)$, the CJ-matrix of $\Phi^{\alpha \to \beta}_\gamma$ is
    $$C_{\Phi^{\alpha \to \beta}_\gamma} = \frac{d_{\alpha}}{d_{\gamma}}(Q_{\alpha}^{\frac{1}{2}}\otimes {\rm Id}_{\beta})p^{\overline{\alpha},\beta}_{\overline{\gamma}}(Q_{\alpha}^{\frac{1}{2}}\otimes {\rm Id}_{\beta}),$$
where $p^{\overline{\alpha},\beta}_{\overline{\gamma}}$ is the orthogonal projection onto $H_{\overline{\gamma}}$ from $H_{\bar \alpha}\otimes H_{\beta}\cong \bar{H_{\alpha}}\otimes H_{\beta}$.
\end{thm}

\begin{proof}

Let $C$ be the CJ-matrix of $\Phi^{\alpha \to \beta}_\gamma$. Then Theorem \ref{thm-cov-Choi}, Proposition \ref{prop-intertwiner} and Proposition \ref{prop-CG-map-cov} tell us that  $(Q_{\alpha}^{-\frac{1}{2}}\otimes {\rm Id})C (Q_{\alpha}^{-\frac{1}{2}}\otimes {\rm Id})$ is a linear combination of orthogonal projections $p^{\overline{\alpha},\beta}_{\delta}$, where $\delta \subseteq \overline{\alpha}\tp \beta$, $\delta \in {\rm Irr}(\QG)$. In order for the desired conclusion we will prove that $C(Q_{\alpha}^{-\frac{1}{2}}\otimes {\rm Id})v^{\overline{\alpha},\beta}_{\delta}=0$ for any $\delta\neq \overline{\gamma}$ in ${\rm Irr}(\g)$, which is the same as
    $$0 = (v^{\overline{\alpha},\beta}_{\delta} )^*(Q_{\alpha}^{-\frac{1}{2}}\otimes {\rm Id})C\otimes 1 =  (u^{\delta})^*u^{\delta}\left[(v^{\overline{\alpha},\beta}_{\delta} )^*(Q_{\alpha}^{-\frac{1}{2}}\otimes {\rm Id})C\otimes 1\right] = (u^{\delta})^*X.$$
The element $X = u^{\delta}\left[(v^{\overline{\alpha},\beta}_{\delta} )^*(Q_{\alpha}^{-\frac{1}{2}}\otimes {\rm Id})C\otimes 1\right]$ in the above is a linear combination of the matrix coefficients of the representation $\delta$. We will show that they can be written as a linear combination of the matrix coefficients of the representation $\overline{\gamma}$. First we note
\begin{align*}
X &= ((v^{\overline{\alpha},\beta}_{\delta})^*\otimes 1) (u^{\overline{\alpha}} \tp u^{\beta})((Q_{\alpha}^{-\frac{1}{2}}\otimes {\rm Id})C\otimes 1)\\
&=((v^{\overline{\alpha},\beta}_{\delta})^*(Q_{\alpha}^{\frac{1}{2}}\otimes {\rm Id})\otimes 1) ((u^{\alpha})^c \tp u^{\beta})((Q_{\alpha}^{-1}\otimes {\rm Id})C\otimes 1)
\end{align*}
Moreover, we have
\begin{align*}
&((u^{\alpha})^c \tp u^{\beta})( (Q_{\alpha}^{- 1}\otimes {\rm Id})C\otimes 1)\\
&=\sum_{i,j=1}^{n_{\alpha}} ((u^{\alpha})^c)_{13}(u^{\beta})_{23} (Q_{\alpha}^{-1}e^\alpha_{ij}\otimes \Phi^{\alpha \to \beta}_\gamma(e^\alpha_{ij})  \otimes 1)\\
&=\sum_{i,j} ((u^{\alpha})^c)_{13}(u^{\beta})_{23}({\rm id}\otimes {\rm id}\otimes {\rm Tr}_{Q_\gamma}\otimes {\rm id})(Q_{\alpha}^{-1}e^\alpha_{ij}\otimes v^{\beta,\gamma}_{\alpha}e^\alpha_{ij}(v^{\beta,\gamma}_{\alpha})^* \otimes 1)\\
&=\sum_{i,j} {\rm Tr}^3_{Q_\gamma} \left ( ((u^{\alpha})^c)_{14}(u^{\beta})_{24} \left [Q_{\alpha}^{-1}e^\alpha_{ij}\otimes v^{\beta,\gamma}_{\alpha}e^\alpha_{ij}(v^{\beta,\gamma}_{\alpha})^* \otimes 1 \right ]\right )\\
& =\sum_{i,j} {\rm Tr}^3_{Q_\gamma} \left ( ((u^{\alpha})^c)_{14}(u^{\beta})_{24} (u^{\gamma})_{34} \left [Q_{\alpha}^{-1}e^\alpha_{ij}\otimes v^{\beta,\gamma}_{\alpha}e^\alpha_{ij}(v^{\beta,\gamma}_{\alpha})^* \otimes 1 \right ] \left  (u^{\gamma}\right )_{34}^*\right ) \\
& =\sum_{i,j} {\rm Tr}^3_{Q_\gamma} \left ( ((u^{\alpha})^c)_{14} \left [ Q_{\alpha}^{-1}e^\alpha_{ij}\otimes (v^{\beta,\gamma}_{\alpha}\otimes 1)u^{\alpha}(e^\alpha_{ij}(v^{\beta,\gamma}_{\alpha})^*\otimes 1) \right ]  \left  (u^{\gamma}\right )_{34}^*\right ) \\
& =\sum_{i,j} {\rm Tr}^3_{Q_\gamma} \left (  ({\rm id}\otimes v^{\beta,\gamma}_{\alpha}\otimes 1)((u^{\alpha})^c)_{13}(u^{\alpha})_{23} \left [ Q_{\alpha}^{-1}e^\alpha_{ij}\otimes e^\alpha_{ij}(v^{\beta,\gamma}_{\alpha})^*\otimes 1 \right ]  \left  (u^{\gamma}\right )_{34}^*\right ) \\
& =\sum_{i,j} {\rm Tr}^3_{Q_\gamma} \left (  ({\rm id}\otimes v^{\beta,\gamma}_{\alpha}\otimes 1) \left [ Q_{\alpha}^{-1}e^\alpha_{ij}\otimes e^\alpha_{ij}(v^{\beta,\gamma}_{\alpha})^*\otimes 1 \right ]  \left  (u^{\gamma}\right )_{34}^*\right ).
\end{align*}
Here, we write ${\rm Tr}^3_{Q_\gamma} = ({\rm id}\otimes {\rm id}\otimes {\rm Tr}_{Q_\gamma}\otimes {\rm id})$ for simplicity. The above fourth equality is due to (1) of Lemma \ref{lem1} and last equality is obtained by (2) of Lemma \ref{lem1}. Now Schur orthogonality tells us that $X$ must be zero unless $\delta = \overline{\gamma}$, the conclusion we wanted.

The last step is to determine the scalar $\lambda\in \Comp$ satisfying
    $$(Q_{\alpha}^{-\frac{1}{2}}\otimes {\rm Id})C(Q_{\alpha}^{-\frac{1}{2}}\otimes {\rm Id})=\lambda p^{\overline{\alpha},\beta}_{\overline{\gamma}}.$$
Note that $v^{\overline{\alpha},\beta}_{\overline{\gamma}}Q_{\gamma}^{-1}=(Q_{\alpha}^{-1}\otimes Q_{\beta})v^{\overline{\alpha},\beta}_{\overline{\gamma}}$ implies $v^{\overline{\alpha},\beta}_{\overline{\gamma}}Q_{\gamma}=(Q_{\alpha}\otimes Q_{\beta}^{-1})v^{\overline{\alpha},\beta}_{\overline{\gamma}}$. Thus, we have
\begin{align*}
d_{\alpha}&={\rm Tr}\left (\frac{d_{\alpha}}{d_{\beta}}Q_{\beta}^{-1}\right )={\rm Tr}(\Phi({\rm Id}_{n_{\alpha}})Q_{\beta}^{-1})=({\rm Tr}\otimes {\rm Tr}_{Q_{\beta}^{-1}})(C)\\
&=\lambda ({\rm Tr}\otimes {\rm Tr})((Q_{\alpha}\otimes Q_{\beta}^{-1})p^{\overline{\alpha},\beta}_{\overline{\gamma}})=\lambda{\rm Tr}(v^{\overline{\alpha},\beta}_{\overline{\gamma}}Q_{\gamma} (v^{\overline{\alpha},\beta}_{\overline{\gamma}})^*)=\lambda d_{\gamma},
\end{align*}
which is the conclusion we wanted.
\end{proof}

Combining all the above we can determine the set of $\QG_{(\alpha,\beta)}$-covariant linear maps and characterize all UCP maps inside of it. The case of CPTP $\g_{(\alpha,\beta)}$-covariant maps is more involved, but a similar conclusion holds if $\g$ is of Kac type.

\begin{thm}\label{thm-classification}
Suppose that the irreducible decomposition $\overline{\alpha}\tp \beta \cong \oplus_{j=1}^n \overline{ \gamma_j}$ is multiplicity free. Then we have the following:
    \begin{enumerate}
        \item The set $\{\Phi^{\alpha \to \beta}_\gamma: \overline{\gamma} \subseteq \overline{\alpha}\tp \beta\}$ of CG-maps is a basis for the linear space $\displaystyle {\rm Cov}_\QG(\alpha,\beta)$.
        
        \item $\displaystyle {\rm CPCov}_\QG(\alpha,\beta) = \left\{\sum_{j=1}^n a_j \Phi^{\alpha \to \beta}_{\gamma_j}: a_j\ge 0,\; 1\le j \le n\right\}$.
        
        \item $\displaystyle {\rm Ext}({\rm UCPCov}_\QG(\alpha,\beta)) = \left\{\frac{d_\beta}{d_\alpha}\Phi^{\alpha \to \beta}_\gamma: \overline{\gamma} \subseteq \overline{\alpha}\tp \beta\right\}$.
        
        \item $\displaystyle {\rm Ext}({\rm CPTPCov}_\QG(\alpha,\beta)) = \{\Phi^{\alpha \to \beta}_\gamma: \overline{\gamma} \subseteq \overline{\alpha}\tp \beta\}$ if, in addition, $\QG$ is of Kac type.
    \end{enumerate}
\end{thm}
\begin{proof}
\begin{enumerate}
\item Theorem \ref{thm-classification} states that the CJ-matrix of $\Phi = \sum_{j=1}^n a_j \Phi^{\alpha \to \beta}_{\gamma_j}$ is
\[ C_\Phi = \sum_{j=1}^n a_j \cdot \frac{d_{\alpha}}{d_{\gamma}}(Q_{\alpha}^{\frac{1}{2}}\otimes {\rm Id}_{\beta})p^{\overline{\alpha},\beta}_{\overline{\gamma_j}}(Q_{\alpha}^{\frac{1}{2}}\otimes {\rm Id}_{\beta}),\]
and any $u^c\tp v$-invariant operator should be of the above form by Theorem \ref{thm-cov-Choi} and Proposition \ref{prop-intertwiner}. Moreover, orthogonality between the associated projections $p^{\overline{\alpha},\beta}_{\overline{\gamma}}$ is transferred to linear independence of the set of CG-maps $\Phi^{\alpha\rightarrow \beta}_{\gamma}$.
\item Positivity of the coefficients $a_j$ follows from orthogonality between the associated projections $p^{\overline{\alpha},\beta}_{\overline{\gamma_j}}$.
\item The above (2) says that any element in ${\rm CPCov}_{\g}(\alpha,\beta)$ is written as $\Phi = \sum_{j=1}^n \frac{b_jd_{\beta}}{d_{\alpha}} \Phi^{\alpha \to \beta}_{\gamma_j}$ with $b_j\geq 0$ for all $j$, and $\Phi$ is unital iff $\displaystyle\sum_{j=1}^n b_j=1$ thanks to Proposition \ref{prop-CG-map-cov} (3). This means that ${\rm CPCov}_{\g}(\alpha,\beta)$ is the set of convex combinations of $\Phi^{\alpha \to \beta}_{\gamma_j}$ and, moreover, the CG-maps are actually all extreme points due to their linear independence.
\item Note that $\Phi^{\alpha\rightarrow \beta}_{\gamma}$ is trace-preserving if $\g$ is of Kac type. Thus, the set ${\rm CPTPCov}_{\g}(\alpha,\beta)$ is $\left \{\sum_{j=1}^n a_j \Phi^{\alpha \to \beta}_{\gamma_j}: \sum_j a_j=1\text{ with }a_j\ge 0\text{ for all }j\right\},$
and the linear independence of $\left \{\Phi^{\alpha \to \beta}_{\gamma_j}:1\leq j\leq n\right\} $ gives us the conclusion we wanted.
\end{enumerate}
\end{proof}

\begin{rem}
The space ${\rm Cov}_{\g}(\alpha,\beta)$ was studied for the following cases:
\begin{itemize}
    \item $G$ is finite group and $\alpha=\beta$ \cite{MSD17}
    \item $G=SU(2)$ \cite[Corollary 4.6, Proposition 5.1]{AN14}
\end{itemize}
\end{rem}

\section{Trace duality and Clebsch-Gordan maps}

Recall that for a linear map $\Phi: B(H_A) \to B(H_B)$ we have $\Phi$ is CP $\Leftrightarrow$ $\Phi^*$ is CP $\Leftrightarrow$ $\Phi'$ is CP and $\Phi$ is TP $\Leftrightarrow$ $\Phi^*$ is unital $\Leftrightarrow$ $\Phi'$ is unital. Here, $\Phi^*$ and $\Phi'$ are the adjoint maps introduced in Section \ref{subsec-Choi-adjoint}.

We first observe that quantum group covariance transfers to the adjoint maps.


\begin{prop}\label{prop-cov-adjoint}
Let $\QG$ be a compact quantum group of Kac type, then a linear map $\Phi: B(H_\alpha) \to B(H_\beta)$, $\alpha,\beta \in {\rm Irr}(\QG)$, is $\QG_{(\alpha,\beta)}$-covariant if and only if  $\Phi^*$ is $\QG_{(\beta, \alpha)}$-covariant.
\end{prop}
\begin{proof}
The $\QG_{(u,v)}$-covariance of $\Phi$ is equivalent to
    $$\sum^{n_u}_{i,j,k,l = 1}\Phi(e^u_{ij} X e^u_{lk}) \otimes u_{ij}u^*_{kl} = \sum^{n_v}_{p,q,r,s = 1}e^v_{pq}\Phi(X)e^v_{sr} \otimes v_{pq}v^*_{rs}$$
for any $X\in B(H_u)$. Applying the map ${\rm Tr}(\cdot~ Y) \otimes {\rm id}$ on both sides for any $Y\in B(H_v)$ we get
    $$\sum_{i,j,k,l}e^u_{lk}\Phi^*(Y)e^u_{ij} \otimes u_{ij}u^*_{kl} = \sum_{p,q,r,s}\Phi^*(e^v_{sr} Y e^v_{pq}) \otimes v_{pq}v^*_{rs}.$$
Now we apply ${\rm id} \otimes S$ on both sides to get the desired conclusion. Here, $S$ is the antipode map and we need the Kac type condition for $S(u^*_{kl}) = u_{lk}$. 
\end{proof}

Now we would like to focus on the adjoint maps of CG-maps, especially in the Kac type case.

\begin{lem}\label{lem-adjoint-maps}
Let $\alpha,\beta,\gamma \in {\rm Irr}(\g)$. The adjoint map $\left (\Phi^{\alpha \to \beta}_{\gamma}\right )':B(H_{\beta})\rightarrow B(H_{\alpha})$ is given by
\begin{align}
(\Phi^{\alpha \to \beta}_{\gamma})'(X)&= (v^{\beta,\gamma}_{\alpha})^t (X\otimes Q_{\gamma})\overline{(v^{\beta,\gamma}_{\alpha})}\label{eq-adjoint1}\\
&=\frac{d_{\alpha}}{d_{\gamma}}Q_{\alpha}^{\frac{1}{2}}({\rm id}\otimes {\rm Tr})(p^{\overline{\alpha},\beta}_{\overline{\gamma}}({\rm Id}_{\alpha}\otimes X^t))Q_{\alpha}^{\frac{1}{2}} \label{eq-adjoint2}
\end{align}
In particular, we have $\displaystyle (\Phi^{\alpha \to \beta}_\gamma)^*(Q_{\beta}) = (\Phi^{\alpha \to \beta}_\gamma)'(Q_{\beta}) =Q_{\alpha}$.
\end{lem}
\begin{proof}
Let $C$ be the CJ-matrix of $\Phi^{\alpha \to \beta}_{\gamma}$. For any $Y\in B(H_\beta)$ we have
    \begin{align*}
        \llangle (\Phi^{\alpha \to \beta}_{\gamma})'(X), Y \rrangle
        &= {\rm Tr}(X^t \Phi^{\alpha \to \beta}_{\gamma}(Y))\\
        &= {\rm Tr}(X^t ({\rm Tr}\otimes {\rm id})(C(Y^t\otimes {\rm Id}))\\
        &= ({\rm Tr} \otimes {\rm Tr})\left (\frac{d_\alpha}{d_\gamma} p^{\overline{\alpha},\beta}_{\overline{\gamma}} (Q^{\frac{1}{2}}_\alpha Y^t Q^{\frac{1}{2}}_\alpha \otimes X^t) \right )\text{  by }(\ref{eq-CJ})\\
        &= \frac{d_\alpha}{d_\gamma} ( {\rm Tr} \otimes {\rm Tr})((Y^t\otimes {\rm Id}) (Q^{\frac{1}{2}}_\alpha \otimes X^t)p^{\overline{\alpha},\beta}_{\overline{\gamma}} (Q^{\frac{1}{2}}_\alpha \otimes {\rm Id}) ),
    \end{align*}
which means that
    \begin{align*}
        (\Phi^{\alpha \to \beta}_{\gamma})'(X)
        &= \frac{d_\alpha}{d_\gamma} ({\rm id} \otimes {\rm Tr})( (Q^{\frac{1}{2}}_\alpha \otimes X^t)p^{\overline{\alpha},\beta}_{\overline{\gamma}} (Q^{\frac{1}{2}}_\alpha \otimes {\rm Id}) )\\
        &= \frac{d_\alpha}{d_\gamma} Q^{\frac{1}{2}}_\alpha \left[ ({\rm id} \otimes {\rm Tr})(p^{\overline{\alpha},\beta}_{\overline{\gamma}} ({\rm Id}\otimes X^t)) \right] Q^{\frac{1}{2}}_\alpha
    \end{align*}
The last statement is directly from \eqref{Q-matrix-v} and  \eqref{eq-adjoint1}.
\end{proof}

\begin{lem}\label{lem-proj-flip}
For any $\alpha,\beta,\gamma$ with $\gamma\subseteq \alpha\tp \beta$ we have
\begin{equation}
(Q_{\beta}^{\frac{1}{2}}\otimes Q_{\alpha}^{\frac{1}{2}})\circ \Sigma\circ \overline{v^{\alpha,\beta}_{\gamma}}Q_{\gamma}^{-\frac{1}{2}}
\end{equation}
is an intertwiner between $\overline{\gamma}$ and $\overline{\beta}\tp \overline{\alpha}$, where $\Sigma: H_{\alpha}\otimes H_{\beta}\rightarrow H_{\beta}\otimes H_{\alpha},\; \xi\otimes \eta\mapsto \eta\otimes \xi$ is the flip map.
\end{lem}
\begin{proof}
Let us simply write $v^{\alpha,\beta}_{\gamma} = v$. We begin with the intertwining property $(v\otimes 1)u^{\gamma} = (u^{\alpha}\tp u^{\beta})(v\otimes 1)$, which can be written as
    \begin{align*}
        \sum^{n_\gamma}_{i,j=1}ve^\gamma_{ij} \otimes u^\gamma_{ij}
        &= \sum^{n_\alpha}_{p,q=1}\sum^{n_{\beta}}_{r,s=1} (e^\alpha_{pq}\otimes e^\beta_{rs})v \otimes u^\alpha_{pq}u^\beta_{rs} \in B(H_{\alpha}\otimes H_{\beta})\otimes C(\g).
    \end{align*}
Applying transpose on $B(H_{\alpha}\otimes H_{\beta})$ and then applying the involution map $*$ on $B(H_{\alpha}\otimes H_{\beta})\otimes C(\g)$ we get
\begin{align}\label{ineq1}
(\overline{v}\otimes 1)(u^{\gamma})^c &= \sum_{p,q,r,s}(e^{\alpha}_{pq}\otimes e^{\beta}_{rs})\overline{v}\otimes (u^{\beta}_{rs})^*(u^{\alpha}_{pq})^*.
\end{align}
Finally,
we apply $\Sigma \otimes {\rm id}$ to both sides of (\ref{ineq1}) to get
\begin{align*}
(\Sigma \circ \overline{v}\otimes 1)(u^{\gamma})^c 
&= \sum_{p,q,r,s}\Sigma \circ (e^{\alpha}_{pq}\otimes e^{\beta}_{rs})\overline{v}\otimes (u^{\beta}_{rs})^*(u^{\alpha}_{pq})^*\\
&= \sum_{p,q,r,s} (e^{\beta}_{rs}\otimes e^{\alpha}_{pq})\circ \Sigma \circ \overline{v}\otimes (u^{\beta}_{rs})^*(u^{\alpha}_{pq})^*\\
&=\left [\sum_{p,q,r,s} e^{\beta}_{rs}\otimes e^{\alpha}_{pq}\otimes (u^{\beta}_{rs})^*(u^{\alpha}_{pq})^*\right ]\cdot (\Sigma \circ  \overline{v}\otimes 1)\\
&=((u^{\beta})^c\tp (u^{\alpha})^c) (\Sigma \circ  \overline{v}\otimes 1),
\end{align*}
which leads us directly to the conclusion we wanted.
\end{proof}

Now the following theorem tells us that the category of $\g$-covariant maps is closed under taking the adjoint map if $\g$ is of Kac type. A special case where $G=SU(2)$ was studied in \cite[Corollary 6.3]{AN14}.

\begin{thm}\label{thm-Kac-ajoint}
Let $\g$ be of Kac type and suppose that $\alpha,\beta,\gamma\in {\rm Irr}(\g)$ such that $\overline{\gamma}\subseteq \overline{\alpha}\tp \beta$ in a multiplicity-free way. Then the adjoint map of $\Phi^{\alpha \to \beta}_\gamma$ is given by
\begin{equation}
(\Phi^{\alpha \to \beta}_\gamma)^*=\frac{d_{\alpha}}{d_{\beta}} \cdot \Phi^{\beta \to \alpha}_{\overline{\gamma}}.
\end{equation}
\end{thm}
\begin{proof}
For any $X\in B(H_\alpha)$ we have
\begin{align*}
(\Phi^{\alpha \to \beta}_\gamma)'(X)&=\frac{d_{\alpha}}{d_{\gamma}}({\rm id}\otimes {\rm Tr})(p^{\overline{\alpha},\beta}_{\overline{\gamma}}({\rm Id}_{\alpha}\otimes X^t))\\
&=\frac{d_{\alpha}}{d_{\gamma}}({\rm Tr}\otimes {\rm id})(\sigma(p^{\overline{\alpha},\beta}_{\overline{\gamma}})(X^t\otimes {\rm Id}_{\alpha})),
\end{align*}
where $\sigma:B(H_{\alpha}\otimes H_{\beta})\rightarrow B(H_{\alpha}\otimes H_{\beta})$ is the flip map given by $\sigma(a\otimes b)=\Sigma \circ (a\otimes b)\circ \Sigma = b\otimes a$ for the flip map $\Sigma$ on $H_\alpha\otimes H_\beta$. Lemma \ref{lem-proj-flip}, together with the Kac type condition, tells us that $(p^{\overline{\beta},\alpha}_{\gamma})^t = \overline{p^{\overline{\beta},\alpha}_{\gamma}} = \sigma(p^{\overline{\alpha},\beta}_{\overline{\gamma}})$, so that we have
\[ (\Phi^{\alpha \to \beta}_\gamma)^*(X)=\frac{d_{\alpha}}{d_{\gamma}}({\rm Tr}\otimes {\rm id})( p^{\overline{\beta},\alpha}_{\gamma}(X^t\otimes {\rm Id}_{\alpha})).\]
This is exactly the CJ-matrix of the linear map $\displaystyle \frac{d_{\alpha}}{d_{\beta}} \cdot \Phi^{\beta \to \alpha}_{\overline{\gamma}}$ by Theorem \ref{thm-Choi}, so that we get the desired conclusion.
\end{proof}

\begin{rem}

Proposition \ref{prop-cov-adjoint} and Theorem \ref{thm-Kac-ajoint} tell us that, if $\g$ is of Kac type and $\alpha,\beta\in {\rm Irr}(\g)$, then ${\rm CPTPCov}_{\g}(\alpha,\beta)$ and ${\rm UCPCov}_{\g}(\beta,\alpha)$ are equivalent in the sense that there is a linear bijection $\Psi: B(H_\alpha, H_\beta) \to B(H_\beta, H_\alpha),\; \Phi \mapsto \frac{d_{\alpha}}{d_{\beta}}\Phi^*$ such that $\Psi({\rm CPTPCov}_{\g}(\alpha,\beta)) = {\rm UCPCov}_{\g}(\beta,\alpha)$.

\end{rem}

\section{Examples}

In this section we demonstrate that Theorems \ref{thm-Choi} and \ref{thm-classification} serve as a method of determining all $\QG_{(\alpha,\beta)}$-covariant maps, consisting of two steps, through several concrete examples. The first step is to find all components $\overline{\gamma}$ in the irreducible decomposition of $\overline{\alpha}\tp \beta$ to get the associated orthogonal projections $p^{\overline{\alpha},\beta}_{\overline{\gamma}}$, which traces back to the extreme points $\Phi^{\alpha\rightarrow \beta}_{\gamma}$ due to Theorem \ref{thm-Choi}. This procedure is quite handy, and was already used in \cite{K02,H18} for some special cases without mentioning why the idea should work. We know by now that multiplicity-free assumption guarantees the success of the procedure thanks to Theorem \ref{thm-Choi}.

What Theorem \ref{thm-Choi} is also telling us is the details of the Stinespring procedure for the resulting channel. More precisely, we need to find an isometric intertwiner $v^{\beta,\gamma}_{\alpha}\in \text{Hom}(\alpha,\beta\tp \gamma)$, which leads us to the corresponding CG-channels $\Phi^{\alpha\rightarrow \beta}_{\gamma}=({\rm id}\otimes {\rm Tr}_{Q_\gamma})(v^{\beta,\gamma}_{\alpha}\cdot (v^{\beta,\gamma}_{\alpha})^*)$. This is the second step for determining all $\QG_{(\alpha,\beta)}$-covariant maps.

The symmetry (quantum) groups that we are going to focus on are the quantum permutation group $S_n^+$ and the $q$-deformed quantum group $SU_q(2)$, which unearth certain quantum phenomena. Then, we will consider abelian groups through projective representations allowing us to revisit the Weyl-covariant channels and explain how to obtain an analogous covariance with respect to Majorana operator in the fermionic system.

\color{black}


\subsection{The permutation groups versus the quantum permutation groups}\label{sec-free-general}

In this section we would like to compare covariances with respect to $S_n$ and $S^+_n$ for $n \ge 4$. Note that $S_2 = S^+_2 = \z_2$, the abelian group with two points and $S_3 = S^+_3$, the simplest non-abelian group. This means that $n\ge 4$ is a natural restriction for observing genuine quantum phenomena.

Let us recall basics of representation theory of $S_n$. We know that ${\rm Irr}(S_n)$ can be identified with the set of all partitions of $n$:
    $$\{\lambda = (\lambda_1, \cdots, \lambda_k): \lambda_1\ge \cdots \ge \lambda_k,\; \lambda_j \in \n,\; \sum^k_{j=1}\lambda_j = n\}.$$
Note that the fundamental representation $\pi: S_n \to B(\ell^2_n)$ given by $\pi(\tau)(e_k) = e_{\tau(k)}$, $1\le k\le n$, has the irreducible decomposition $u^{(n)}\oplus u^{(n-1,1)}$. Indeed $\pi$ has two invariant subspaces $H_0=\Comp\cdot \xi_0$ with $\xi_0=\displaystyle \sum_{j=1}^n e_j$ and $W=H_0^{\perp}$, which correspond to the irreducible representations $u^{(n)}$ and $u^{(n-1,1)}$ respectively. Here $u^{(n)}$ is the trivial representation, and we write $u^{(n-1,1)}$ simply by $V$. Immediate consequences are ${\rm dim}(V) = n-1$ and $V$ is self-conjugate, i.e. $\overline{V} = V$.



From \cite[p.384, Example]{Tok84} we have a multiplicity-free decomposition
    \begin{equation}\label{eq-perm-decomp}
        (n-1,1) \tp (n-1,1) \cong (n) \oplus (n-1,1) \oplus (n-2,2) \oplus (n-2,1,1)
    \end{equation}
with ${\rm dim}(n-2,2) = \frac{n(n-3)}{2}$ and ${\rm dim}(n-2,1,1) = \frac{(n-1)(n-2)}{2}$ from the dimension formula \cite[(4.11)]{FulHar91}. 
Thus, we get realizations of the representations $(n-2,2)$ and $(n-2,1,1)$ once we find invariant subspaces of $(n-1,1) \tp (n-1,1)$ with the corresponding dimensions.

Now we turn our attention to the case of $S^+_n$, whose representation category has the same fusion rule as $SO(3)$. Recall that ${\rm Irr}(S_n^+)$ is identified with $\left \{u^0,u^1,u^2,\cdots\right\}$. Let us write $u^1$ simply by $U$. Then $U\tp U$ has the following irreducible decomposition
    $$U \tp U \cong u^2 \oplus u^1 \oplus 1$$
and the dimension of $u^2$ is $n^2-3n+1$. Note that the permutation group $S_n$ can be understood as a (closed) quantum subgroup of $S^+_n$, i.e. there is a surjective unital $*$-homomorphism $R: C(S^+_n) \to C(S_n)$ such that $\Delta_{S_N} R = (R\otimes R)\Delta_{S_n^+}$. Here, the maps $\Delta_{S_N}$ and $\Delta_{S_n^+}$ are co-multiplications on $S_n$ and $S_n^+$, respectively. Actually the map $R$ is determined by $R(U_{ij})=V_{ij}$ for all $1\leq i,j\leq N-1$. Moreover, the map ${\rm id} \otimes R$ sends an arbitrary unitary representation on $S^+_n$ into a unitary representation on $S_n$. In particular, we have 
$({\rm id} \otimes R)(U^2) = (n-2,2) \oplus (n-2,1,1)$ since
\begin{align*}
    1\oplus V \oplus ({\rm id}\otimes R)(u^2)&= ({\rm id}\otimes {\rm id}\otimes R)(U\tp U)\\
    & =V\tp V\\
    & =1\oplus V \oplus (n-2,2)\oplus (n-2,1,1).
\end{align*} 
Moreover, applying ${\rm id} \otimes R$ to \eqref{eq-bipartite-inv} we can see that $(S^+_n)_{(U,U)}$-invariant states are automatically $(S_n)_{(V,V)}$-invariant. 


Combining all the above observations we can conclude that for $n\ge 4$ the convex set of all $(S_n)_{V,V}$-covariant channels has exactly four extreme points $\Phi_j$, $1\le j\le 4$ with $\Phi_1 = id$, $\Phi_2 = \Phi^{V\to V}_{V}$, $\Phi_3 = \Phi^{V\to V}_{(n-2,2)}$, $\Phi_4 = \Phi^{V\to V}_{(n-2,1,1)}$. Meanwhile, the convex set of all $(S^+_n)_{U,U}$-covariant channels has exactly three extreme points, $\Phi_1$, $\Phi_2$ and $\frac{n(n-3)}{2(n^2-3n+1)}\Phi_3 + \frac{(n-1)(n-2)}{2(n^2-3n+1)} \Phi_4$. Indeed, if we start with $(S_n^+)_{U,U}$-invariant states, they are spanned by the orthogonal projections $p^{U,U}_1=p^{V,V}_1$, $p^{U,U}_U=p^{V,V}_V$ and $p^{U,U}_{u^2}=p^{V,V}_{(n-2,2)}+p^{V,V}_{(n-2,1,1)}$. Then, by Theorem \ref{thm-Choi}, they trace back to $\Phi_1$, $\Phi_2$ and $\frac{n(n-3)}{2(n^2-3n+1)}\Phi_3 + \frac{(n-1)(n-2)}{2(n^2-3n+1)} \Phi_4$ respectively. The geometric picture of this $2$-simplex inside the tetrahedron with the vertices $\Phi_1,\cdots,\Phi_4$ is given as follows. This finishes the first step for both of the (quantum) groups. 

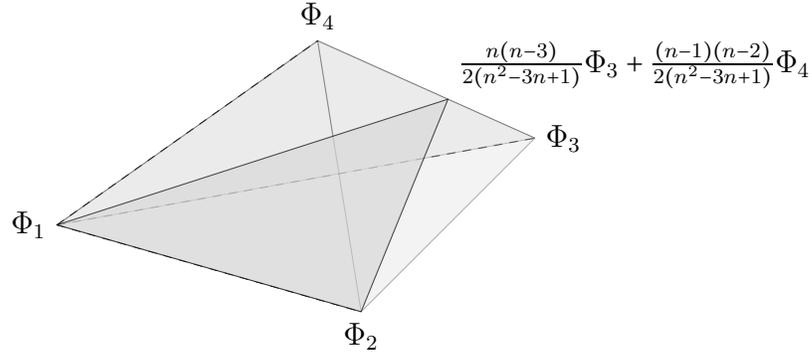
\color{black}
\begin{flushright}
\begin{tikzpicture}[line join = round, line cap = round]
\pgfmathsetmacro{\factor}{1};
\coordinate [label=left:$\Phi_1$] (A) at (0,0,0*\factor);
\coordinate [label=below:$\Phi_2$] (B) at (5.1962,0,3*\factor);
\coordinate [label=right:$\Phi_3$] (C) at (5.1962,0,-3*\factor);
\coordinate [label=above:$\Phi_4$] (D) at (3.4641,2.4495,0*\factor);
\coordinate [label=above right :$\frac{n(n-3)}{2(n^2-3n+1)}\Phi_3+\frac{(n-1)(n-2)}{2(n^2-3n+1)}\Phi_4$] (E) at (4.5034,0.9798,-1.8*\factor);

\foreach \i in {A,B,C,D}
    \draw[dashed] (0,0)--(\i);
\draw[-, fill=gray!30, opacity=.3] (A)--(D)--(B)--cycle;
\draw[-, fill=gray!30, opacity=.3] (A)--(D)--(C)--cycle;
\draw[-, fill=gray!30, opacity=.3] (B)--(D)--(C)--cycle;
\draw[-, fill=gray!30, opacity=.7] (A)--(B)--(E)--cycle;
\end{tikzpicture}
\end{flushright}


\subsubsection{Revisiting the case of $S_4$ and $S^+_4$}\label{sec-free-special}

Let us continue to the second step in the case of $n=4$, namely, figuring out the Stinespring representations of the associated CG-channels.

Recall the fundamental representation $\pi: S_4 \to B(\ell^2_4)$, $\pi(\tau)(e_k) = e_{\tau(k)}$, $1\le k\le 4$. We choose a new ONB $\{f_k: 0\le k\le 3\}$ given by     
    $$\begin{cases}
f_0 = \frac{1}{2}(e_1+e_2+e_3+e_4),\\
f_1 = \frac{1}{2}(e_1-e_2+e_3-e_4),\\
f_2 = \frac{1}{2}(e_1-e_2-e_3+e_4),\\
f_3 = \frac{1}{2}(e_1+e_2-e_3-e_4).
\end{cases}$$
Then, $f_0$ is an $\pi$-invariant vector and the subspace $W = \{f_0\}^\perp = {\rm span}\{f_k:1\le k \le 3\}$ clearly has an ONB $\{f_k:1\le k \le 3\}$, which will be our choice of basis below.

Now we would like to find realizations of the representations $(2,2)$ and $(2,1,1)$ through the irreducible decomposition
    $$V\tp V \cong (4) \oplus (3,1) \oplus (2,2) \oplus (2,1,1)$$
for the representation $V = (3,1) = \pi|_W$. First, we record the matrix form of $V$ (with respect to the basis $\{f_k = |k\ra :1\le k \le 3\})$ in the case of generators (transpositions) $(12), (23), (34)$ of $S_4$, which is
    {\small \begin{equation}\label{eq-V-matrix}
    V(12) = \begin{bmatrix}0&-1&0\\ -1&0&0\\ 0&0&1\end{bmatrix},\; V(23) = \begin{bmatrix}0&0&1\\ 0&1&0\\ 1&0&0\end{bmatrix},\; V(34) = \begin{bmatrix}0&1&0\\ 1&0&0\\ 0&0&1\end{bmatrix}.
    \end{equation}
    }
For the orthonormal vectors $\begin{cases}
|\psi\ra = \frac{1}{\sqrt{3}}(|11\ra + |22\ra + |33\ra),\\
|g^\pm_1\ra = \frac{1}{\sqrt{2}}(|23\ra \pm |32\ra),\\
|g^\pm_2\ra = \frac{1}{\sqrt{2}}(|31\ra \pm |13\ra),\\
|g^\pm_3\ra = \frac{1}{\sqrt{2}}(|12\ra \pm |21\ra)
\end{cases}$ and the unit vectors $\begin{cases}
|h_1\ra = \frac{1}{\sqrt{6}}(-2|11\ra +|22\ra + |33\ra),\\
|h_2\ra = \frac{1}{\sqrt{6}}(|11\ra -2|22\ra + |33\ra),\\
|h_3\ra = \frac{1}{\sqrt{6}}(|11\ra +|22\ra -2 |33\ra),
\end{cases}$
we can readily check that 4 irreducible $V$-invariant subspaces $\begin{cases}
W_1 = \Comp |\psi\ra,\\
W_2 = {\rm span}\{g^+_1,g^+_2,g^+_3\},\\
W_3 = {\rm span}\{h_1,h_2,h_3\},\\
W_4 = {\rm span}\{g^-_1,g^-_2,g^-_3\}
\end{cases}$ correspond to the components $(4), (3,1), (2,2)$ and $(2,1,1)$. Note that ${\rm dim}(W_3) = 2$ and we have chosen non-orthogonal vectors for later use.

In order to figure out the CG-channels coming from the associated orthogonal projections $p^{V,V}_{(4)},p^{V,V}_{(3,1)},p^{V,V}_{(2,2)},p^{V,V}_{(2,1,1)}$ we need to look at the relations $V \subseteq V \tp 1$, $V \subseteq V \tp V$, $V \subseteq V \tp (2,2)$ and $V \subseteq V \tp (2,1,1)$ and find the corresponding intertwining isometries $v_1, \cdots, v_4$.


The first case $v_1: \Comp^3 \to \Comp^3 \otimes \Comp \cong \Comp^3$ is nothing but the identity map, i.e. $v_1 = id_{\Comp^3}$, which means that $\Phi^{V\rightarrow V}_{(4)}={\rm id}_3$.

Secondly, we set $v_2: \Comp^3 \to \Comp^3\otimes W_2\subseteq \Comp^3\otimes \Comp^3 \otimes \Comp^3,\;\; |k\ra \mapsto |g^+_k\ra$. Then, it is straightforward to check that $v_2$ is the wanted intertwiner, i.e. $(V(\tau)\otimes V(\tau))v_2 = v_2 \circ V(\tau)$ for each $\tau \in S_4$ using the matrix form \eqref{eq-V-matrix}. Note that it actually is enough to consider the cases $\tau = (12), (23), (34)$. Then, the resulting channel is {\small
$$\Phi^{V\rightarrow V}_{(3,1)}:\left [ \begin{array}{ccc} a_{11}&a_{12}&a_{13}\\a_{21}&a_{22}&a_{23}\\a_{31}&a_{32}&a_{33} \end{array} \right ]\mapsto \frac{1}{2}\left [ \begin{array}{ccc} a_{22}+a_{33}&a_{21}&a_{31}\\a_{12}&a_{11}+a_{33}&a_{32}\\a_{13}&a_{23}&a_{11}+a_{22} \end{array} \right ].$$
}

For $v_3$ and $v_4$ the situation is a bit more complicated since our understanding of the representations $(2,2)$ and $(2,1,1)$ is an indirect one coming from the decomposition \eqref{eq-perm-decomp} or, in other words, as restrictions of $V\tp V$. We first consider $v_3: \Comp^3 \to \Comp^3 \otimes W_3\subseteq \Comp^3 \otimes \Comp^3 \otimes \Comp^3$ given by $v_3(|k\ra) := |k\ra \otimes |h_k\ra$. Then, we can readily check that $v_3$ is the intertwiner we were looking for, i.e. $(V(\tau)\otimes V(\tau) \otimes V(\tau))v_3 = v_3 \circ V(\tau)$ for each $\tau \in S_4$. Now we define $v_4: \Comp^3 \to \Comp^3 \otimes W_4\subseteq \Comp^3 \otimes \Comp^3 \otimes \Comp^3$ by
$\begin{cases}
v_4(|1\ra) := \frac{1}{\sqrt{2}}(|2 g^-_3\ra - |3 g^-_2\ra),\\
v_4(|2\ra) := \frac{1}{\sqrt{2}}(|3 g^-_1\ra - |1 g^-_3\ra),\\
v_4(|3\ra) := \frac{1}{\sqrt{2}}(|1 g^-_2\ra - |2 g^-_1\ra).
\end{cases}$
We can similarly check that $v_4$ is the intertwiner for $V\subseteq V \tp (2,1,1)$, i.e. $(V(\tau)\otimes V(\tau) \otimes V(\tau))v_4 = v_4 \circ V(\tau)$ for each $\tau \in S_4$. Finally, we get the resulting channels:
{\small
$$\Phi^{V\rightarrow V}_{(2,2)}:\left [ \begin{array}{ccc} a_{11}&a_{12}&a_{13}\\a_{21}&a_{22}&a_{23}\\a_{31}&a_{32}&a_{33} \end{array} \right ]\mapsto \frac{1}{2}\left [ \begin{array}{ccc} 2a_{11}&-a_{12}&-a_{13}\\-a_{21}&2a_{22}&-a_{23}\\-a_{31}&-a_{32}&2a_{33} \end{array} \right ],$$
$$\Phi^{V\rightarrow V}_{(2,1,1)}:\left [ \begin{array}{ccc} a_{11}&a_{12}&a_{13}\\a_{21}&a_{22}&a_{23}\\a_{31}&a_{32}&a_{33} \end{array} \right ]\mapsto \frac{1}{2}\left [ \begin{array}{ccc} a_{22}+a_{33}&-a_{21}&-a_{31}\\-a_{12}&a_{11}+a_{33}&-a_{32}\\-a_{13}&-a_{23}&a_{11}+a_{22} \end{array} \right ].$$
}


Now we move to the case of quantum permutation group $S_4^+$. From the discussion in the beginning of Section \ref{sec-free-general} we know that $U\tp U$ has only three invariant subspaces which correspond to the following orthogonal projections/quantum channels
$$
\begin{array}{lllll}
    p^{1,1}_0&=p^{V,V}_{(n)}&\leftrightarrow & \Phi^{1\rightarrow 1}_0&=\Phi^{V\rightarrow V}_{(4)}={\rm id}_3\\
    p^{1,1}_1&=p^{V,V}_{(3,1)}&\leftrightarrow &\Phi^{1\rightarrow 1}_1&=\Phi^{V\rightarrow V}_{(3,1)} \\
    p^{1,1}_2&=p^{V,V}_{(2,2)}+p^{V,V}_{(2,1,1)}&\leftrightarrow & \Phi^{1\rightarrow 1}_2&=\frac{2}{5}\Phi^{V\rightarrow V}_{(2,2)}+\frac{3}{5}\Phi^{V\rightarrow V}_{(2,1,1)}
\end{array}.
$$

We only need to determine the Stinespring representation of the channel $\Phi^{1 \to 1}_2$. Let $w$ be the isometric intertwiner for $u^1 \subseteq  u^1\tp u^2$, then by applying ${\rm id}\otimes R$ we can see that $w$ is also an intertwiner for $V\subseteq (V \tp(2,2)) \oplus (V \tp (2,1,1))$. Since the multiplicity of $V$ in $(V \tp (2,2)) \oplus (V \tp (2,1,1))$ is 2, we know that $w = \alpha v_3 + \beta v_4:\Comp^3\rightarrow \Comp^3\otimes (W_3 + W_4)$ for some $\alpha, \beta \in \Comp$. 
Comparing with the identity $\Phi^{1\rightarrow 1}_2=\frac{2}{5}\Phi^{V\rightarrow V}_{(2,2)}+\frac{3}{5}\Phi^{V\rightarrow V}_{(2,1,1)}$ we actually get $w = \sqrt{\frac{2}{5}}v_3 + \sqrt{\frac{3}{5}} v_4$.

\begin{rem}
We have seen the difference between $(S_N)_{(V,V)}$-covariance and $(S_n^+)_{(U,U)}$-covariance by counting the number of extreme points for each case. We can apply the same idea for other free quantum groups, namely free orthogonal quantum groups $O_N^+$ and free unitary quantum groups $U_N^+$ containing $O_N$ and $U_N$ as closed quantum subgroups, respectively. Let us denote the fundamental representations of $O_N$, $O_N^+$, $U_N$, $U_N^+$ by $v$, $V$, $u$ and $U$ respectively. Here, the structures of $(v,v)$-, $(u,u)$-, $(u,\overline{u})$-covariant quantum channels have been studied in \cite{VW01,K02,H18}, and the number of their extreme points are $3,2,2$ respectively. In the quantum group perspective, their analogous notions should be $(V,V)$-, $(U,U)$, $(U,\overline{U})$-covariant quantum channels, and the number of their extreme points are $2,2,1$ respectively. 

\end{rem}

\subsection{The case of $q$-deformed quantum group $SU_q(2)$}\label{sec-SUq(2)}

\subsubsection{Qubit channels with $SU_q(2)$-covariance}\label{sec-SUq(2)-special}
The representation theory for $SU_q(2)$ is parallel with the one of $SU(2)$, so that we have ${{\rm Irr}(SU_q(2))} = \{0,1,2,\cdots,\}$, where $0$ corresponds to the trivial representation. See \cite{KK89,KS97} for details about the Clebsch-Gordan coefficients of $SU_q(2)$. Let us determine the set of all $SU_q(2)_{(1,1)}$-covariant channels, which are qubit channels. From the decomposition $\overline{1}\tp 1 \cong 2 \oplus 0$ and the fact that $\overline{k} \cong k$ we know that we need to focus on $1\subseteq 1\tp 0$, $1\subseteq 1\tp 2$ and the associated isometric intertwiners $v^{1,0}_1$, $v^{1,2}_1$, which are given by $v^{1,0}_1 = id_{\Comp^2}$ and
\begin{align*}
    v^{1,2}_1|0\ra&= \frac{1}{\sqrt{1+q^2+q^4}}|01\ra-\frac{\sqrt{q^2+q^4}}{\sqrt{1+q^2+q^4}}|10\ra\\
        v^{1,2}_1|1\ra&= \frac{\sqrt{1+q^2}}{\sqrt{1+q^2+q^4}}|02\ra-\frac{q^2}{\sqrt{1+q^2+q^4}}|11\ra.
\end{align*}
Then, the resulting CG-maps are $\Phi^{1\rightarrow 1}_0={\rm id}_2$ and
    $$\Phi^{1\rightarrow 1}_2\left ( \left [ \begin{array}{cc}a&b\\c&d\end{array} \right ] \right )= \frac{1}{1+q^2+q^4}\left [ \begin{array}{cc} a+d(q^2+q^4)&-bq^2\\-cq^2&a(1+q^2)+dq^4\end{array} \right ].$$
Thus, we get
    $${\rm UCPCov}_{SU_q(2)}(1,1)=\left \{ p\cdot \Phi^{1\rightarrow 1}_0 +(1-p)\cdot \Phi^{1\rightarrow 1}_2: 0\leq p\leq 1\right\},$$
the convex set of all $SU_q(2)_{(1,1)}$-covariant quantum channels in the Heisenberg picture.    

\begin{rem}
In the Schr{\" o}dinger picture we have a completely different conclusion. Indeed, a $SU_q(2)_{(1,1)}$-covariant linear map $a\cdot \Phi^{1\rightarrow 1}_0 +b\cdot \Phi^{1\rightarrow 1}_2$ is trace-preserving if and only if $a=1$ and $b=0$. Thus, the identity map ${\rm id}_2=\Phi^{1\rightarrow 1}_0$ is the only trace-preserving $SU_q(2)_{(1,1)}$-covariant map. 

Surprisingly, the same conclusion still holds for $SU_q(2)_{(k,k)}$ with arbitrary $k\in \n$. In other words, the identity map ${\rm id}_{k+1}=\Phi^{k\rightarrow k}_0$ is the unique $SU_q(2)_{(k,k)}$-covariant CPTP map. Moreover, we can prove that there exists no $SU_q(2)_{(k,l)}$-covariant CPTP map if $k>l$. See Corollary \ref{cor-SUq(2)} for the details.
\end{rem}


\begin{rem}\label{rmk-quantum-trace}
The use of quantum trace for the definition of CG-maps is essential to get $\QG$-covariance. Indeed, if we use the usual trace together with the isometry $v^{1,2}_1$ as in \cite{BCLY20}, then we have
\begin{align*}
    &({\rm id}\otimes {\rm Tr})\left ( v^{1,2}_1 \left [\begin{array}{cc} a&b\\c&d \end{array}\right ](v^{1,2}_1)^* \right )\\
    &=\frac{1}{1+q^2+q^4}\left [\begin{array}{cc} a+(1+q^2)d&-q^2b\\-q^2c&a(q^2+q^4)+dq^4 \end{array}\right ],
\end{align*}
which is trace-preserving. However, it is not a linear combination of $\Phi^{1\rightarrow 1}_0$ and $\Phi^{1\rightarrow 1}_2$.
\end{rem}

\begin{rem}
Note that, for a compact group $G$, the complementary channels of $G$-Clebsch-Gordan channels 
$$({\rm Tr}\otimes {\rm id})(v^{\beta,\gamma}_{\alpha}\cdot (v^{\beta,\gamma}_{\alpha})^*)$$
are always $G_{(\alpha,\gamma)}$-covariant. This is not automatically reproduced in the quantum group case. Indeed, a natural choice of ``complementary map" of $\Phi^{1\rightarrow 2}_1$ would be
$$\rho\mapsto ({\rm Tr}_{Q_2}\otimes {\rm id})(v^{2,1}_1\rho (v^{2,1}_1)^*),$$ 
which is quantum trace-preserving, but is not $SU_q(2)_{(1,1)}$-covariant. More precisely, the above map is
$$\left [ \begin{array}{cc}a&b\\c&d\end{array} \right ]\mapsto  \frac{1}{1+q^2+q^4}\left [ \begin{array}{cc} aq^4+d(q^4+q^6)&-bq^2\\-cq^2&a(q^{-2}+1)+d\end{array} \right ],$$
which is not a linear combination of $\Phi^{1\rightarrow 1}_0$ and $\Phi^{1\rightarrow 1}_2$.
\end{rem}

\subsection{$SU_q(2)$-covariant CP maps are rarely TP}\label{sec-SUq(2)-existence}

We begin with a result applicable for a general quantum group $\QG$.

\begin{prop}\label{prop-general-main}
Suppose that $\alpha,\beta\in {\rm Irr}(\g)$ and $\displaystyle \overline{\alpha}\tp \beta$ has a multiplicity-free irreducible decomposition.
\begin{enumerate}
\item There exists a CPTP $\QG_{(\alpha, \beta)}$-covariant map only if 
\[ \frac{n_{\alpha}\norm{Q_{\alpha}}}{d_{\alpha}} \leq \frac{n_{\beta}\norm{Q_{\beta}}}{d_{\beta}}.\]
\item Any CPTP $\QG_{(\alpha, \alpha)}$-covariant map is a convex combination of $\Phi^{\alpha\to \alpha}_\gamma$ for which $\overline{\gamma}\subseteq \overline{\alpha}\tp\alpha$ and $Q_\gamma= I_\gamma$.
\end{enumerate}
\end{prop}

\begin{proof}
\begin{enumerate}
\item Let $\Phi:B(H_{\alpha})\rightarrow B(H_{\beta})$ be a CPTP $\QG_{(\alpha, \beta)}$-covariant map and let $\overline{\alpha}\tp \beta\cong \oplus_{j=1}^n \overline{\gamma_j}$ is the multiplicity-free irreducible decomposition. Then we can write $\displaystyle \Phi = \sum_{j=1}^n a_j \Phi^{\alpha \to \beta}_{\gamma_j}$, and we have $\Phi^*(Q_{\beta})=\displaystyle \left ( \sum_{j=1}^n a_j \right )Q_{\alpha}$ by Lemma \ref{lem-adjoint-maps}. Since $\Phi^*$ is a UCP map, it is contractive, so that we have
    $$\frac{\norm{Q_{\beta}}}{\norm{Q_{\alpha}}} \geq \sum_{j=1}^n a_j.$$
Moreover, (3) of Proposition \ref{prop-CG-map-cov} tells us that $\Phi(I_\alpha) = (\sum_j a_j)\frac{d_\alpha}{d_\beta}I_\beta$, so that we have $$\displaystyle \sum_{j=1}^n a_j= \frac{n_{\alpha}d_{\beta}}{d_{\alpha}n_{\beta}}$$ by trace preserving property of $\Phi$. Thus the desired conclusion directly follows.

\item Note that the sequence $(a_j)_{j=1}^n$ is now indeed a probability distribution. Moreover, \eqref{eq-adjoint2} tells us that $\displaystyle \sum_{j=1}^n \frac{a_j}{d_{\gamma_j}}({\rm id}\otimes {\rm Tr})(p^{\overline{\alpha},\beta}_{\overline{\gamma_j}})=\frac{1}{d_{\alpha}}Q_{\alpha}^{-1}$, so that we get $\displaystyle \sum_{j=1}^n a_j \frac{n_{\gamma_j}}{d_{\gamma_j}}=1$ by taking trace on both sides. Now we can easily conclude that $a_j=0$ whenever $\displaystyle \frac{n_{\gamma_j}}{d_{\gamma_j}}<1$.
\end{enumerate}
\end{proof}

The above theorem covers a broad class of non-Kac compact quantum groups. For example, we can demonstrate that $SU_q(2)$-covariant CP maps are rarely trace-preserving as follows.

\begin{cor}\label{cor-SUq(2)}
Let $\g=SU_q(2)$ with $\displaystyle 0<q<1$. Recall that ${\rm Irr}(SU_q(2))$ can be identified with $\left \{0,1,2,\cdots\right\}$.
\begin{enumerate}
\item For any $k,l\in {\rm Irr}(SU_q(2))$ with $k>l$, there is no CPTP $SU_q(2)_{(k,l)}$-covariant map.
\item The only CPTP $SU_q(2)_{(k,k)}$-covariant map is the identity map on $M_{k+1}$.
\end{enumerate}
\end{cor}
\begin{proof}
\begin{enumerate}
\item By (1) of Proposition \ref{prop-general-main} it is enough to show that the following function $k\mapsto \displaystyle \frac{n_k \norm{Q_k}}{d_k}=\displaystyle \frac{(1-q^2)\cdot (k+1)}{1-q^{2(k+1)}}$ is strictly increasing. Indeed, for $f(x)=\displaystyle \frac{x+1}{1-q^{2(x+1)}}$, we have $\displaystyle\frac{f'(x)}{f(x)}=\frac{1}{x+1}+\frac{q^{2(x+1)}\cdot \log(q^2)}{1-q^{2(x+1)}}$
and $f'(x)>0$ is equivalent to the condition
\begin{equation}\label{eq-cond1}
    1-q^{2(x+1)}+(x+1)q^{2(x+1)}\log(q^2) > 0.
\end{equation}
Now the Taylor expansion $q^{-2(x+1)}=\displaystyle \sum_{k=0}^{\infty}\frac{(\log(q^{-2}))^k}{k!}(x+1)^k$ shows that $q^{-2(x+1)} > 1+(x+1)\log(q^{-2}),$
which is the same as \eqref{eq-cond1}.

\item Since $\overline{u^{k}}\tp u^k \cong u^0\oplus u^2\oplus \cdots \oplus u^{2k}$, the desired conclusion comes from the fact that $Q_{n}\neq {\rm Id}$ for all $n\neq 0$.
\end{enumerate}
\end{proof}

\begin{rem}
The case for $k<l$ is still open, and there are some possibilities to find CPTP $SU_q(2)_{(k,l)}$-covariant maps. Indeed, the following map
\[\lambda\mapsto \frac{\lambda}{d+1}{\rm Id}_{d+1}\]
is clearly a CPTP $SU_q(2)_{(0,d)}$-covariant channel for any $d\in \n$.
\end{rem}

\subsection{Covariance with respect to projective representations}\label{sec-proj}

From the beginning, the study of group symmetry allowed not only the representations of groups, but also {\em projective representations} coming from {\em 2-cocycle twistings}. For a compact group $G$ we say that a continuous function $\sigma: G\to \tor$ is {\em 2-cocycle} if $\sigma(s,t)\sigma(st,u)=\sigma(s,tu)\sigma(t,u)$ and $\sigma(s,e) = \sigma(e,t)=1$ for all $s,t,u\in G$. A unitary {\em projective representation} of $G$ with respect to $\sigma$ (simply, $\sigma$-representation) is a map (continuous under the strong operator topology) $\pi: G \to \mc U(H_\pi)$ satisfying $\pi(s)\pi(t) = \sigma(s,t)\pi(st)$. This gives us a natural action of $G$ to the states via conjugation, and consequently we get $G$-invarince of states and $G$-covariance of channels with respect to projective representations as in the ordinary representation case.

In a sense this new set of tools is not a big surprise since the theory of projective representation is closely related to ordinary representation theory. More precisely, the above projective representation $\pi$ can be lifted to a central extension $G_\sigma := \tor\times G$ of $G$ with the group law $(s,x)\cdot (t,y):=(st\sigma(x,y), xy)$. The actual lifting is the map $\tilde{\pi}: G_\sigma \to \mc U(H_\pi),\; (s,x) \mapsto s\pi(x)$, which becomes an ordinary unitary representation of the compact group $G_\sigma$. Note that $\tilde{\pi}$ is known to be irreducible if and only if $\pi$ is irreducible. One can easily see that the action of $G$ through $\pi$ and of $G_\sigma$ through $\tilde{\pi}$ is exactly the same.

However, it is still valuable to separate the case of projective representations due to their intimate connection to fundamental quantum systems such as (finite) Weyl systems and fermionic systems. This case will provide us further non-trivial examples of multiplicity-free irreducible decompositions, and lead us back to the well-known concept of {\em Weyl covariant channels} as we will cover below.


From now on we would like to narrow our attention to the case that $G$ is a finite abelian group, where we can actually find nontrivial examples, and we will use additive notation for the group law (i.e. $x+y$, instead of $xy$) of $G$. The group $G$ is equipped with a 2-cocycle $\sigma$, and we will
require that the map $G \times G \to \tor,\; (x,y) \to \sigma(x,y)\overline{\sigma(y,x)}$ gives rise to a group isomorphism between $G$ and the dual group $\widehat{G}$ of $G$ consisting of characters on $G$.
All the above assumptions on $G$ and $\sigma$ guarantee that there exists only one (upto unitary equivalence) irreducible $\sigma$-representation $W: G \to B(H_W)$ \cite{DigernesVaradarajan04}.

The limited supply of irreducible $\sigma$-representations of $G$ forces us to focus on $G_{(W,W)}$-covariant linear map
    $$\Phi: B(H_W) \to B(H_W),$$
which means that $\Phi(W(x)AW(x)^*) = W(x)\Phi(A)W(x)^*$, $x\in G$, $A \in B(H_W)$.    
Note that it is easy to check that the adjoint map $\Phi^*$ of a $G_{(W,W)}$-covariant map $\Phi$ is again $G_{(W,W)}$-covariant, which means that we could still stay in the Schr{\" o}dinger picture in this case. In order to determine the structure of all $G_{(W,W)}$-covariant linear maps we can use exactly the same argument as before as long as the representation $\overline{W}\tp W$ has a multiplicity-free irreducible decomposition. One notable difference here is that $\overline{W}\tp W$ is an ordinary representation of $G$ while $W$ is a projective one. Moreover, it is well known that irreducible unitary representations of $G$ are nothing but the characters on $G$, namely the elements of $\widehat{G}$. Thus, we need to check that $\overline{W}\tp W \cong \bigoplus^N_{j=1}\gamma_j$, where all the elements $\gamma_j \in \widehat{G}$, $1\le j \le N$ are distinct. This actually happens in the following examples.

\subsubsection{Finite Weyl systems}
Let $G = F\times \widehat{F}$, where $F$ is another finite abelian group, equipped with the 2-cocyle $\sigma((x,\gamma), (y,\delta)) := \gamma(y)$. We recall the translation and the modulation operator $T_x$ and $M_\gamma$ for $x\in F$, $\gamma \in \widehat{G}$ given by
    $$T_xf(u) : =f(u-x),\;\;  M_\gamma f(u):=\gamma(u)f(u),\;\; f\in L^2(F),u\in F.$$
Then, the unique $\sigma$-representation $W:G \to B(H_W)$ is given by $H_W=L^2(F)$ and
    $$W(\xv) := T_x M_\gamma,\;\; \xv=(x,\gamma) \in G.$$
Our choice of the Haar measure on $F$ is the counting measure, so that the family
$\{\delta_x: x\in F\}$ is an orthonormal basis of $H_W=L^2(F)$, where $\delta_x$ is the Dirac delta function at $x\in F$. 
Now we consider a family $\{f_\xv\}_{\xv = (x,\gamma)\in G}$ of orthonormal basis for $H_W\otimes H_W = L^2(F\times F)$ given by
    \begin{equation}\label{eq-new-basis}
    f_\xv = |F|^{-\frac{1}{2}}\sum_{y\in F}\delta_y \otimes W(\xv)\delta_y.    
    \end{equation}
For $\yv = (x',\gamma')\in G$ we have
    \begin{align*}
        [\overline{W}(\yv) \otimes W(\yv)]f_\xv
        & = |F|^{-\frac{1}{2}}\sum_{y\in F}\overline{W(\yv)\delta_y} \otimes W(\yv)W(\xv)\delta_y\\
        & = \overline{\sigma(\xv,\yv)}\sigma(\yv,\xv)|F|^{-\frac{1}{2}}\sum_{y\in F}\overline{W(\yv)\delta_y} \otimes W(\xv)W(\yv)\delta_y\\
        & = \overline{\sigma(\xv,\yv)}\sigma(\yv,\xv)|F|^{-\frac{1}{2}}\sum_{y\in F}\delta_{x'+y} \otimes W(\xv)\delta_{x'+y}\\
        & = \overline{\sigma(\xv,\yv)}\sigma(\yv,\xv)f_\xv,
    \end{align*}
where we are using the fact that $W(\yv)\delta_y = \gamma'(y)\delta_{x'+y}$. This means that the representation $\overline{W}\tp W$ has a multiplicity free irreducible decomposition
    $$\overline{W}\tp W = \sum_{\xv \in G}\varphi_\xv |f_\xv\ra \la f_\xv|,$$
where $\varphi_\xv$ is a character on $G$ given by $\varphi_\xv(\yv) = \overline{\sigma(\xv,\yv)}\sigma(\yv,\xv)$, $\yv \in G$. In other words, we have
    $$\overline{W}\tp W \cong \bigoplus_{\xv \in G}\varphi_\xv.$$
Finally, we can easily check that the unitary channel ${\rm Ad}_{W(\xv)}$, $\xv\in G$ corresponds to the minimal projection $|f_\xv\ra \la f_\xv|$, i.e. $C_{{\rm Ad}_{W(\xv)}} = |F| \cdot |f_\xv\ra \la f_\xv|$.

Now we obtain the following result, which was known for the case of $F=\z_d$, $d\in N$, a.k.a. Weyl covariant channels \cite[Theorem 4]{SC18}. Note that the corresponding invariant states was already investigated in \cite[Example 6]{VW01}.

\begin{thm}
We have $\displaystyle {\rm Ext}({\rm CPTPCov}_G(W,W)) = \{{\rm Ad}_{W(\xv)}: \xv \in G\}$. In other words, any $G_{(W,W)}$-covariant channel $\Phi$ is of the form
    $$\Phi(A) = \sum_{(x,\gamma)\in G}p_\xv W(\xv) A W(\xv)^*,\;\; A\in B(H_W)$$
for some probability distribution $(p_\xv)_{\xv\in G}$.
\end{thm}

\subsubsection{Fermionic system}
The fermionic system in $n$-modes can be described by the group $G = F\times \widehat{F}$, where $F = \z^n_2$. The difference from the finite Weyl system on the same group $G$ is that we use the following 2-cocycle.
    $$\sigma_{\text{fer}}(\xv, \yv) := (-1)^{\xv^t \Delta \yv},\; \xv,\yv \in G \cong \z^{2n}_2, \;\text{where} \;\Delta = {\Tiny \begin{bmatrix} 0 &&&\\ 1& 0&&\\1&1&0&\\ \vdots& \vdots& \ddots & \ddots & \\1 & 1 & \cdots & 1 & 0 \end{bmatrix}}.$$
The unique irreducible unitary $\sigma_{\text{fer}}$-representation  $W = W_{\text{fer}}:G \cong \z^{2n}_2 \to B(\ell^2(F))= M_{2^n}(\Comp)$ is given by $$W_{\text{fer}}(\xv) := \ch^{x_1}_1 \cdots \ch^{x_{2n}}_{2n},\;\; \xv = (x_1,\cdots,x_{2n}) \in \z^{2n}_2.$$
Here, we use the {\em Majorana operators} $\ch_1 , \dots, \ch_{2n}$, which are self-adjoint operators acting on $\ell^2(F)$ satisfying the CAR: 
    $$\{\ch_j, \ch_k\} = 2\delta_{jk},\; 1\le j,k\le 2n.$$
Note that $\ch_j$'s are identified with
\begin{align*}
    \ch_{2j-1}=Z\otimes\cdots\otimes Z\otimes X\otimes I\otimes \cdots \otimes I\\
    \ch_{2j}=Z\otimes\cdots\otimes Z\otimes Y\otimes I\otimes \cdots \otimes I,
\end{align*}
where
    $$X={\footnotesize \begin{bmatrix} 0&1\\1&0\end{bmatrix}},\,Y={\footnotesize \begin{bmatrix} 0&-i\\i&0\end{bmatrix}},\,Z={\footnotesize\begin{bmatrix} 1&0\\0&-1\end{bmatrix}},$$
the usual Pauli matrices for qubit, and we have $X$ and $Y$ at $j$-th tensor component in the above.

Now we would like to consider $(\z^{2n}_2)_{(W_{\text{fer}}, W_{\text{fer}})}$-covariant channels, which we should call {\em fermionic covariant} $2n$-qubit channels. We may apply the same argument as in the finite Weyl system case. More precisely, we consider an ONB $\{f_\xv\}_{\xv\in G}$ as in \eqref{eq-new-basis}. A careful look at the action of $W(\xv)$ on the vector $\delta_y$, $y\in F$, says that there is a permutation $\tau: F\to F$ and a constant $D_{\xv,y} \in \{\pm\}$ such that $W(\xv)\delta_y = D_{\xv,y}\delta_{\tau(y)}$. Indeed, for $\xv = (x_1,x_2,\cdots,x_{2n-1},x_{2n})$ and $y=(y_1,\cdots,y_n)$ we have
\begin{align*}
\lefteqn{W(\xv)\delta_y}\\
& = \hat{c}^{x_1}_1 \hat{c}^{x_2}_2\cdots \hat{c}^{x_{2n-1}}_{2n-1} \hat{c}^{x_{2n}}_{2n} \delta_y\\
& = (X^{x_1}Y^{x_2}\otimes I \otimes I \otimes \cdots)(Z^{x_3+x_4}\otimes X^{x_3}Y^{x_4} \otimes I \otimes \cdots)\\
& \;\;\;\; \cdots (\delta_{y_1}\otimes \cdots \otimes \delta_{y_n})\\
& = (-1)^{y_1(x_3 + \cdots + x_{2n})}X^{x_1}Y^{x_2} \delta_{y_1} \otimes (-1)^{y_2(x_5 + \cdots + x_{2n})}X^{x_3}Y^{x_4} \delta_{y_2} \otimes \cdots,
\end{align*}
which explains the above claim.

Now we can repeat the same argument after \eqref{eq-new-basis} to get the following.
\begin{thm}
We have $\displaystyle {\rm Ext}({\rm CPTPCov}_{\z^{2n}_2}(W_{\text{fer}}, W_{\text{fer}})) = \{{\rm Ad}_{W_{\text{fer}}(\xv)}: \xv \in G\}$. In other words, any fermionic covariant channel $\Phi$ is of the form
    $$\Phi(A) = \sum_{(x,\gamma)\in G}p_\xv W_{\text{fer}}(\xv) A W_{\text{fer}}(\xv)^*,\;\; A\in B(H_W)$$
for some probability distribution $(p_\xv)_{\xv\in G}$.
\end{thm}

\emph{Acknowledgements}: S-G. Youn was funded by the New Faculty Startup Fund from Seoul National University and by the National Research Foundation of Korea (NRF) grant funded by the Korea government (MSIT) (No. 2020R1C1C1A01009681). H.H. Lee was supported by the Basic Science Research Program through the National Research Foundation of Korea (NRF) Grant NRF-2017R1E1A1A03070510 and the National Research Foundation of Korea (NRF) Grant funded by the Korean Government (MSIT) (Grant No.2017R1A5A1015626).

\bibliographystyle{alpha}
\bibliography{LY20}

\end{document}